\newif\iflncs
\lncsfalse

\iflncs

\documentclass{llncs}

\else

\documentclass[a4paper,11pt]{article}
\usepackage{amsthm}
\usepackage{fullpage}

\fi

\usepackage{hyperref}
\usepackage{graphicx}
\graphicspath{{figs/}}
\usepackage[english]{babel}
\usepackage{latexsym}
\usepackage{amssymb}
\usepackage{amsmath}
\usepackage{amsfonts}
\usepackage{ifthen}
\usepackage{algo}
\usepackage{cite} 
\usepackage{microtype} 

\usepackage[mathlines]{lineno}

\newcommand*\patchAmsMathEnvironmentForLineno[1]{%
  \expandafter\let\csname old#1\expandafter\endcsname\csname #1\endcsname
  \expandafter\let\csname oldend#1\expandafter\endcsname\csname end#1\endcsname
  \renewenvironment{#1}%
     {\linenomath\csname old#1\endcsname}%
     {\csname oldend#1\endcsname\endlinenomath}}%
\newcommand*\patchBothAmsMathEnvironmentsForLineno[1]{%
  \patchAmsMathEnvironmentForLineno{#1}%
  \patchAmsMathEnvironmentForLineno{#1*}}%
\AtBeginDocument{%
\patchBothAmsMathEnvironmentsForLineno{equation}%
\patchBothAmsMathEnvironmentsForLineno{align}%
\patchBothAmsMathEnvironmentsForLineno{flalign}%
\patchBothAmsMathEnvironmentsForLineno{alignat}%
\patchBothAmsMathEnvironmentsForLineno{gather}%
\patchBothAmsMathEnvironmentsForLineno{multline}%
}

\iflncs
 \let\qedhere\qed

\else
 \newtheorem{theorem}{Theorem}
 
 \newtheorem{lemma}[theorem]{Lemma}

\fi

\newcommand{\RR}{\ensuremath{\mathbb R}}  
\newcommand{\ZZ}{\ensuremath{\mathbb Z}}  
\newcommand{\WW}{\ensuremath{\mathbb W}}  
\newcommand{\II}{\ensuremath{\mathbb I}}  
\newcommand{\JJ}{\ensuremath{\mathbb J}}  
\renewcommand{\SS}{\ensuremath{\mathbb S}}  

\renewcommand{\H}{\ensuremath{\mathcal{H}}}  

\def\Var{\mathop{\rm Var}} 
     
\def\leftmost{\mathit{Leftmost}}
\def\rightmost{\mathit{Rightmost}}
\def\Active{\mathit{active}}

\def\Index{{\sc Index}}

\newcommand\eps{\varepsilon}

\def\DEF#1{\textbf{\emph{#1}}}

\newcommand{\out}[1]{}

\iflncs
\else

\fi

\title{Interval Selection in the Streaming Model}

\iflncs

\authorrunning{Cabello and P\'erez-Lantero}
\titlerunning{Interval Selection in the Streaming Model}

\author{
  		Sergio Cabello\inst{1}\thanks{
        	Supported by the Slovenian Research Agency, program P1-0297, projects
            J1-4106 and L7-5459; by the ESF EuroGIGA project (project GReGAS) of
            the European Science Foundation. 
            Part of the work was done while visiting Universidad de
            Valpara\'\i so, Chile.
        }
	\and
  		Pablo P\'erez-Lantero\inst{2}\thanks{
        	Supported by project Millennium Nucleus
        	Information and Coordination in Networks ICM/FIC P10-024F (Chile).
        }
}

\institute{
		Department of Mathematics, IMFM, and
        Department of Mathematics, FMF, University of Ljubljana, Slovenia.\\
        \email{sergio.cabello@fmf.uni-lj.si}
	\and  
		Escuela de Ingenier\'ia Civil en Inform\'atica, Universidad de
		Valpara\'{i}so, Chile.\\
        \email{pablo.perez@uv.cl}
}

\else

\author{
  		Sergio Cabello\thanks{Department of Mathematics, IMFM, and
			Department of Mathematics, FMF, University of Ljubljana, Slovenia.
			\texttt{sergio.cabello@fmf.uni-lj.si}
        	Supported by the Slovenian Research Agency, program P1-0297, projects
            J1-4106 and L7-5459; by the ESF EuroGIGA project (project GReGAS) of
            the European Science Foundation. 
            Part of the work was done while visiting Universidad de
            Valpara\'\i so, Chile.}
	\and
  		Pablo P\'erez-Lantero\thanks{Escuela de Ingenier\'ia Civil 
			en Inform\'atica, Universidad de Valpara\'{i}so, Chile.
			\texttt{pablo.perez@uv.cl}
        	Supported by project Millennium Nucleus
        	Information and Coordination in Networks ICM/FIC RC130003 (Chile).
        }
}

\fi

\begin{document}
\maketitle

\begin{abstract}
	A set of intervals is independent when the intervals are pairwise disjoint.
	In the interval selection problem we are given a set $\II$ of intervals
	and we want to find an independent subset of intervals of largest cardinality.
	Let $\alpha(\II)$ denote the cardinality of an optimal solution.
	We discuss the estimation of $\alpha(\II)$ in the streaming model, 
	where we only have one-time, sequential access to the input intervals,
	the endpoints of the intervals lie in $\{ 1,\dots ,n \}$,
	and the amount of the memory is constrained. 

	For intervals of different sizes,
	we provide an algorithm in the data stream model
	that computes an estimate $\hat\alpha$ of $\alpha(\II)$ that,
	with probability at least $2/3$,
	satisfies $\tfrac 12(1-\eps) \alpha(\II) \le \hat\alpha \le \alpha(\II)$.
	For same-length intervals,
	we provide another algorithm in the data stream model 
	that computes an estimate $\hat\alpha$ of $\alpha(\II)$ that,
	with probability at least $2/3$,
	satisfies $\tfrac 23(1-\eps) \alpha(\II) \le \hat\alpha \le \alpha(\II)$.
	The space used by our algorithms is bounded by a polynomial in $\eps^{-1}$ and $\log n$.
	We also show that no better estimations can be achieved using $o(n)$ bits
	of storage.
	
	We also develop new, approximate solutions to the interval selection problem,
    where we want to report a feasible solution, that use $O(\alpha(\II))$ space. 
	Our algorithms for the interval selection problem match the optimal results by 
	Emek, Halld{\'o}rsson and Ros{\'e}n [Space-Constrained Interval Selection, ICALP 2012],
	but are much simpler.	
\end{abstract}

\section{Introduction}

Several fundamental problems have been explored in the data streaming model; 
see~\cite{chakrabartics11,muthu-book} for an overview. In this model 
we have bounds on the amount of available memory, the data arrives 
sequentially, and we cannot afford to look at input data of the past, unless
it was stored in our limited memory. This is effectively equivalent to assuming
that we can only make one pass over the input data.

In this paper, we consider the interval selection problem. 
Let us say that a set of intervals is \DEF{independent} when all the intervals
are pairwise disjoint. 
In the \DEF{interval selection problem}, the input is a set $\II$ of 
intervals and we want to find an independent subset of largest cardinality.
Let us denote by $\alpha(\II)$ this largest cardinality. 
There are actually two different problems: one problem is finding (or approximating) 
a largest independent subset, while the other problem is estimating $\alpha(\II)$.
In this paper we consider both problems in the data streaming model.

There are many natural reasons to consider the interval selection problem in
the data streaming model. Firstly, the interval selection problem appears
in many different contexts and several extensions have been studied; see 
for example the survey~\cite{survey}.

Secondly, the interval selection problem is a natural
generalization of the distinct elements problem: given a data stream of numbers,
identify how many distinct numbers appeared in the stream. The distinct elements
problem has a long tradition in data streams; 
see Kane, Nelson and Woodruff~\cite{optimal} for an optimal algorithm and
references therein for a historical perspective.

Thirdly, there has been interest in understanding graph problems in the data stream
model. However, several problems cannot be solved within the memory constraints
usually considered in the data stream model. 
This leads to the introduction by Feigenbaum et al.~\cite{semi} of the semi-streaming
model, where the available memory is $O(|V|\log^{O(1)} |V|)$, being $V$
the vertex set of the corresponding graph.
Another model closely
related to preemptive online algorithms was considered by 
Halld{\'o}rsson et al.~\cite{indepset}: there is an output buffer where 
a feasible solution is always maintained.

Finally, geometrically-defined graphs provide a rich family of graphs where certain graph
problems may be solved within the traditional model.
We advocate that graph problems should be considered for geometrically-defined
graphs in the data stream model.
The interval selection problem is one such case, since it is exactly 
finding a largest independent set in the intersection graph of the input intervals.

\paragraph{Previous works.}
Emek, Halld{\'o}rsson and Ros{\'e}n~\cite{EmekHR12} consider the interval
selection problem with $O(\alpha(\II))$ space. They provide a $2$-approximation algorithm
for the case of arbitrary intervals and a $(3/2)$-approximation for the case
of proper intervals, that is, when no interval contains another interval. 
Most importantly, they show that no better
approximation factor can be achieved with sublinear space. 
Since any $O(1)$-approximation obviously requires $\Omega(\alpha(\II))$ space, 
their algorithms are optimal. 
They do not consider the problem of estimating $\alpha(\II)$.
Halld{\'o}rsson et al.~\cite{indepset} consider maximum independent set in
the aforementioned online streaming model.
As mentioned before, estimating $\alpha(\II)$ is a generalization of
the distinct elements problems. See Kane, Nelson and Woodruff~\cite{optimal} and
references therein.

\paragraph{Our contributions.}
We consider both the estimation of $\alpha(\II)$ and the interval selection problem, 
where a feasible solution must be produced, 
in the data streaming model. We next summarize our results and put them in context.

\begin{itemize}
	\item[(a)] We provide a $2$-approximation algorithm for the interval selection problem
    	using $O(\alpha(\II))$ space.
    	Our algorithm has the same space bounds and approximation factor
        than the algorithm by Emek, Halld{\'o}rsson and Ros{\'e}n~\cite{EmekHR12},
        and thus is also optimal.
        However, our algorithm is considerably easier to explain,
        analyze and understand.
        Actually, the analysis of our algorithm is nearly trivial. 
        This result is explained in Section~\ref{sec:lis}.
    \item[(b)] We provide an algorithm to obtain a value $\hat\alpha(\II)$ such that 
    	$\tfrac 12 (1-\eps)\alpha(\II)\le \hat\alpha(\II) \le \alpha(\II)$
        with probability at least $2/3$. 
        The algorithm uses $O(\eps^{-5}\log^6 n)$ space for intervals with 
        endpoints in $\{1,\ldots,n\}$. 
        As a black-box subroutine we use a $2$-approximation algorithm 
        for the interval selection problem.
        This result is explained in Section~\ref{sec:slis}.
    \item[(c)] For \emph{same-length} intervals we provide a $(3/2)$-approximation algorithm 
    	for the interval selection problem using $O(\alpha(\II))$ space.
    	Again, Emek, Halld{\'o}rsson and Ros{\'e}n~\cite{EmekHR12}
		provide an algorithm with the same guarantees and give a lower bound showing that
		the algorithm is optimal. We believe that our algorithm is simpler, 
        but this case is more disputable. 
        This result is explained in Section~\ref{sec:lisu}.
    \item[(d)] For \emph{same-length} intervals with endpoints in $\{1,\ldots,n\}$,
    	we show how to find in $O(\eps^{-2}\log(1/\eps)+\log n)$ space an 
		estimate $\hat\alpha(\II)$ such that 
        $\tfrac 23 (1-\eps)\alpha(\II)\le \hat\alpha(\II) \le \alpha(\II)$ 
        with probability at least $2/3$. This algorithm is an adaptation of the new
        algorithm in (c).
        This result is explained in Section~\ref{sec:slisu}.
    \item[(e)] We provide lower bounds showing that the approximation ratios 
    	in (b) and (d) are essentially optimal, if we use $o(n)$ space. 
        Note that the lower bounds of 
        Emek, Halld{\'o}rsson and Ros{\'e}n~\cite{EmekHR12} hold for the interval
        selection problem but not for the estimation of $\alpha(\II)$. 
        We employ a reduction from the one-way randomized communication complexity of \Index.
        Details appear in Section~\ref{sec:lower}.
\end{itemize}

The results in (a) and (c) work in a comparison-based model and we assume that a unit of memory can store an interval.
The results in (b) and (d) are based on hash functions and 
we assume that a unit of memory can store values in $\{1,\dots ,n\}$.
Assuming that the input data, in our case the endpoints of the intervals,  
is from $\{1,\dots, n\}$ is common in the data streaming model.
The lower bounds of (e) are stated at bit level.

It is important to note that estimating $\alpha(\II)$ requires considerably less space
than computing an actual feasible solution with $\Theta(\alpha(\II))$ intervals.
While our results in (a) and (c) are a simplification of the work of Emek et al.,
the results in (b) and (d) were unknown before. 

As usual, the probability of success can be increased to $1-\delta$ using
$O(\log(1/\delta))$ parallel repetitions of the algorithm and choosing
the median of values computed in each repetition.

\section{Preliminaries}
We assume that the input intervals are closed. Our algorithms can be easily
adapted to handle inputs that contain intervals of mixed types:
some open, some closed, and some half-open.

We will use the term `interval' only for the input intervals.
We will use the term `window' for intervals constructed through
the algorithm and `segment' for intervals associated
with the nodes of a segment tree. (This segment tree is explained later on.) 
The windows we consider may be of any type regarding the inclusion of endpoints.

For each natural number $n$, we let $[n]$ be the integer range $\{ 1,\dots,n\}$.
We assume that $0< \eps<1/2$.

\subsection{Leftmost and rightmost interval}
Consider a window $W$ and a set of intervals $\II$. 
We associate to $W$ two input intervals.
\begin{itemize}
	\item The interval $\leftmost(W)$ is, among the intervals of $\II$ contained in $W$, 
		the one with smallest right endpoint. If there are many candidates 
        with the same right endpoint, $\leftmost(W)$ is one with largest 
        left endpoint.
	\item The interval $\rightmost(W)$ is, among the intervals of $\II$ contained in $W$,
		the one with largest left endpoint. If there are many candidates with the
        same left endpoint, $\rightmost(W)$ is one with smallest right endpoint.
\end{itemize}

When $W$ does not contain any interval of $\II$, then $\leftmost(W)$ and $\rightmost(W)$
are undefined.
When $W$ contains a unique interval $I\in \II$, we have $\leftmost(W)=\rightmost(W)=I$.
Note that the intersection of all intervals contained in $W$ is precisely 
$\leftmost(W)\cap \rightmost(W)$.

In fact, we will consider $\leftmost(W)$ and $\rightmost(W)$ with respect to the portion
of the stream that has been treated. 
We relax the notation by omitting the reference to $\II$ or the portion of the stream
we have processed. It will be clear from the context
with respect to which set of intervals we are considering $\leftmost(W)$ and $\rightmost(W)$.

\subsection{Sampling}

We next describe a tool for sampling elements from a stream.
A family of permutations $\H = \{ h: [n]\rightarrow [n] \}$ is 
\DEF{$\eps$-min-wise independent} if 
\[
	\forall X\subset [n] \text{ and } \forall y\in X:
		~~~~~\frac{1-\eps}{|X|}~\le~ \Pr_{h\in \H} \left[ h(y)= \min h(X)\right]~ \le~ \frac{1+\eps}{|X|}.
\]
Here, $h\in \H$ is chosen uniformly at random.
The family of \emph{all} permutations is $0$-min-wise independent. 
However, there is no compact way to specify an arbitrary permutation.
As discussed by Broder, Charikar and Mitzenmacher~\cite{bcm-03}, the results
of Indyk~\cite{indyk-01} can be used to construct a compact, computable 
family of permutations that is $\eps$-min-wise independent.
See \cite{bcfm-00,cdf-01,dm-02} for other uses of $\eps$-min-wise independent permutations.

\begin{lemma}
\label{le:permutations}
	For every $\eps\in(0,1/2)$ and $n>0$
	there exists a family of permutations $\H (n,\eps)= \{ h: [n]\rightarrow [n] \}$
	with the following properties:
	(i) $\H (n,\eps)$ has $n^{O(\log(1/\eps))}$ permutations;
	(ii) $\H (n,\eps)$ is $\eps$-min-wise independent;
	(iii) an element of $\H (n,\eps)$ can be chosen uniformly at random in $O(\log(1/\eps))$ time;
	(iv) for $h\in \H(n,\eps)$ and $x,y\in [n]$, we can decide with 
		$O(\log(1/\eps))$ arithmetic operations whether $h(x)<h(y)$.
\end{lemma}

\begin{proof}
	Indyk~\cite{indyk-01} showed that there exist constants $c_1,c_2>1$ such that,
	for any $\eps>0$ and any family $\H'$ 
	of $c_2\log(1/\eps)$-wise 	
	independent hash functions $[m]\rightarrow [m]$, it holds the following:
	\begin{align*}
	\forall X\subset [m]\text{ with } |X|\leq \frac{\eps m}{c_1} 
		\text{ and } & \forall y\in [m]\setminus X:\\
		&~~ \frac{1-\eps}{|X|+1} ~\le~ \Pr_{h'\in \H'} [h'(y)< \min h'(X)] ~\le~ \frac{1+\eps}{|X|+1}.
	\end{align*}

	Set $m=c_1 n/\eps>n$ and let 
	$\H'=\{ h'\colon [m]\rightarrow [m]\}$ be a family
	of $c_2\log(1/\eps)$-wise independent hash functions.
	Since $n= \eps m/c_1$, the result of Indyk implies that
	\[
	 \forall X\subset [n]\text{ and } \forall y \in [n]\setminus X:
		~~~~\frac{1-\eps}{|X|+1} ~\le~ \Pr_{h'\in \H'} [h'(y)< \min h'(X)] ~\le~ \frac{1+\eps}{|X|+1}.
	\]
	Each hash function $h'\in \H'$ can be used to create a permutation 
    $\widehat{h'}:[n]\rightarrow [n]$:
	define $\widehat{h'}(i)$ as the position of $(h'(i),i)$ in the lexicographic order of 
	$\{ (h'(i),i) \mid  i\in [n] \}$. 
	Consider the set of permutations 
	$\widehat{\H'}=\{ \widehat{h'}:[n]\rightarrow [n] \mid h'\in \H'\}$.
    For each $X\subset [n]$ and $y \in [n]\setminus X$
    we have
    \begin{align*}
    	\frac{1-\eps}{|X|+1} ~\le ~\Pr_{h'\in \H'} \left[h'(y)< \min h'(X)\right]
        ~\le \Pr_{\widehat{h'}\in \widehat{\H'}} \left[\widehat{h'}(y)< \min \widehat{h'}(X)\right] 
    \end{align*}
    and
    \begin{align*}
    	\Pr_{\widehat{h'}\in \widehat{\H'}} \left[\widehat{h'}(y)< \min \widehat{h'}(X)\right] 
        ~&\le ~\Pr_{h'\in \H'} \left[h'(y)\le \min h'(X)\right] \\
        ~&\le ~\Pr_{h'\in \H'} \left[h'(y)< \min h'(X)\right] + \Pr_{h'\in \H'}\left[ h'(y)= \min h'(X)\right]\\
        ~&\le~ \frac{1+\eps}{|X|+1} + \frac{1}{m}~\le~
        \frac{1+2\eps}{|X|+1},
    \end{align*}
    where we have used that $h'(y)= \min h'(X)$ corresponds to a collision
    and $m> n/\eps$.
	We can rewrite this as
	\[
	 \forall X\subset [n]\text{ and } \forall y \in X:
		~~~~\frac{1-\eps}{|X|} ~\le~ \Pr_{\widehat{h'}\in \widehat{\H'}} \left[ \widehat{h'}(y)= \min \widehat{h'}(X)\right] ~\le~ \frac{1+2\eps}{|X|}.
	\]
    Using $\eps/2$ instead of $\eps$ in the discussion, the 
    lower bound becomes $(1-\eps/2)/|X|\ge (1-\eps)/|X|$ and the upper bound
    becomes $(1+\eps)/|X|$, as desired.
	Standard constructions using polynomials over finite fields can be used
	to construct a family $\H'=\{ h'\colon [m]\rightarrow [m]\}$ 
	of $c_2\log(1/\eps)$-wise independent hash functions such that:
	$\H'$ has $n^{O(\log(1/\eps))}$ hash functions;
	an element of $\H'$ can be chosen uniformly at random in $O(\log(1/\eps))$ time;
	for $h'\in \H'$ and $x\in [n]$ we can compute $h'(x)$ using
	$O(\log(1/\eps))$ arithmetic operations.

	This gives an implicit description of our desired set of permutations $\widehat{\H'}$
	satisfying (i)-(iii). 
	Moreover, while computing $\widehat{h'}(x)$ for $h'\in \H'$ is demanding,
	we can easily decide whether $\widehat{h'}(x)< \widehat{h'}(y)$ by computing
	and comparing $(h'(x),x)$ and $(h'(y),y)$.\qedhere	
\end{proof}

Let us explain now how to use Lemma~\ref{le:permutations} to make a (nearly-uniform) random sample. We learned this idea from Datar and Muthukrishnan~\cite{dm-02}.
Consider any fixed subset $X\subset [n]$ and let $\H=\H(n,\eps)$ be the family of 
permutations given in Lemma~\ref{le:permutations}.
An \DEF{$\H$-random} element $s$ of $X$ is obtained
by choosing a hash function $h\in \H$ uniformly at random,
and setting $s=\arg\min \{ h(t)\mid t\in X\}$. 
It is important to note that $s$ is not chosen uniformly at random from $X$.
However, from the definition of $\eps$-min-wise independence we have
\[
	\forall x\in X:~~~~ \frac{1-\eps}{|X|} ~\le~ \Pr [s=x]~\le~ \frac{1+\eps}{|X|}.
\]
In particular, we obtain the following 
\[
	\forall Y\subset X:~~~~\frac{(1-\eps)|Y|}{|X|} ~\le~ \Pr [s\in Y]~\le~ \frac{(1+\eps)|Y|}{|X|}.
\]
This means that, for a fixed $Y$, 
we can estimate the ratio $|Y|/|X|$ using $\H$-random samples from $X$ repeatedly,
and counting how many belong to $Y$.

Using $\H$-random samples has two advantages for data streams with elements from $[n]$. 
Through the stream, we can maintain an $\H$-random sample $s$ of the elements
seen so far. For this, we select $h\in\H$ uniformly at random,
and, for each new element $a$ of the stream, we check whether $h(a)<h(s)$ and update $s$, if needed.
An important feature of sampling in such way is that $s$ is 
almost uniformly at random among those
appearing in the stream, \emph{without} counting multiplicities.
The other important feature is that we select $s$ at its first appearance in the stream. 
Thus, we can carry out any computation that depends on $s$ 
and on the portion of the stream after its first appearance. 
For example, we can count how many times the $\H$-random element $s$ 
appears in the whole data stream.

We will also use $\H$ to make conditional sampling: we select $\H$-random samples
until we get one satisfying a certain property. To analyze such technique, the following
result will be useful.

\begin{lemma}
\label{le:conditional}
	Let $Y\subset X\subset [n]$ and assume that $0<\eps<1/2$.
    Consider the family of permutations $\H=\H(n,\eps)$ from Lemma~\ref{le:permutations},
    and take a $\H$-random sample $s$ from $X$. Then
    \[
    	\forall y\in Y:~~~~ \frac{1-4\eps}{|Y|} ~\le~ \Pr [s=y\mid s\in Y]~\le~ \frac{1+4\eps}{|Y|}.
    \]
\end{lemma}
\begin{proof}
	Consider any $y\in Y$. Since $s=y$ implies $s\in Y$, we have
    \begin{align*}
    	\Pr [s=y\mid s\in Y] ~=~ \frac{\Pr[s=y]}{\Pr[s\in Y]}
        ~\le~ \frac{\frac{1+\eps}{|X|}}{\frac{(1-\eps)|Y|}{|X|}} 
        ~=~ \frac{1+\eps}{1-\eps}\cdot \frac{1}{|Y|} ~\le~ (1+4\eps) \frac{1}{|Y|},
    \end{align*}
    where in the last inequality we used that $\eps<1/2$.
    Similarly, we have
    \begin{align*}
    	\Pr [s=y\mid s\in Y] ~\ge~ \frac{\frac{1-\eps}{|X|}}{\frac{(1+\eps)|Y|}{|X|}} 
        ~=~ \frac{1-\eps}{1+\eps}\cdot \frac{1}{|Y|} ~\ge~ (1-4\eps) \frac{1}{|Y|},
    \end{align*}
which completes the proof\qedhere
\end{proof}

\section{Largest independent subset of intervals}
\label{sec:lis}

In this section we show how to obtain a $2$-approximation to the largest independent
subset of $\II$ using $O(\alpha(\II))$ space.

A set $\WW$ of windows is a \DEF{partition of the real line} if
the windows in $\WW$ are pairwise disjoint and their union is the whole $\RR$.
The windows in $\WW$ may be of different types regarding the inclusion of endpoints.
See Figure~\ref{fig:partition} for an example.

\begin{lemma}
\label{le:2approx}
	Let $\II$ be a set of intervals and
	let $\WW$ be a partition of the real line with the following properties:
	\begin{itemize}
		\item Each window of $\WW$ contains at least one interval from $\II$.
		\item For each window $W\in \WW$, the intervals of $\II$ contained
			in $W$ pairwise intersect.
	\end{itemize}
	Let $\JJ$ be any set of intervals constructed 
	by selecting for each window $W$ of $\WW$
	an interval of $\II$ contained in $W$.
	Then $|\JJ| > \alpha(\II)/2$.
\end{lemma}
\begin{proof}
	Let us set $k=|\WW|$.
	Consider a largest independent set of intervals $\JJ^*\subseteq \II$.
	We have $|\JJ^*|=\alpha(\II)$. Split the set $\JJ^*$ into two sets $\JJ^*_\subset$ 
	and $\JJ^*_\cap$ as follows: $\JJ^*_\subset$ contains the intervals contained in some window
	of $\WW$ and $\JJ^*_\cap$ contains the other intervals.
	See Figure~\ref{fig:partition} for an example. 
	The intervals in $\JJ^*_\cap$ are pairwise disjoint and each of them intersects
	at least two consecutive windows from $\WW$, thus $|\JJ^*_\cap|\le k-1$. 
	Since all intervals contained
	in a window of $\WW$ pairwise intersect, $\JJ^*_\subset$ has at most one 
    interval per window.
	Thus $|\JJ^*_\subset|\le k$. Putting things together we have 
	\[
		\alpha(\II) ~=~ |\JJ^*| ~=~ |\JJ^*_\cap| + |\JJ^*_\subset| ~\le~ k -1 + k ~=~2k-1.
	\]
	Since $\JJ$ contains exactly $|\WW|=k$ intervals, we obtain
	\[
		2\cdot |\JJ| ~=~ 2k ~>~ 2k-1 ~\ge~ \alpha(\II) ,
	\]
    which implies the result.\qedhere
\end{proof}

\begin{figure}
	\centering
	\includegraphics[page=2]{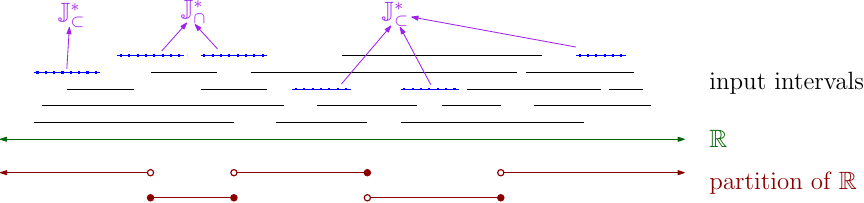}
	\caption{At the bottom there is a partition of the real line. Filled-in disks
		indicate that the endpoint is part of the interval; empty disks
		indicate that the endpoint is not part of the interval.
		At the top, we show the split of some optimal solution $\JJ^*$ (dotted blue)
		into $\JJ^*_\subset$ and $\JJ^*_\cap$.}
	\label{fig:partition}
\end{figure}

We now discuss the algorithm.
Through the processing of the stream, we maintain a partition $\WW$ of the line 
so that $\WW$ satisfies the hypothesis of Lemma~\ref{le:2approx}. 
To carry this out,
for each window $W$ of $\WW$ we store the intervals $\leftmost(W)$ and $\rightmost(W)$.
See Figure~\ref{fig:partition3} for an example.
To initialize the structures, we start with a unique window $\WW = \{ \RR \}$
and set $\leftmost(W)= \rightmost(W)= I_0$, where $I_0$ is the first interval of the
stream. With such initialization, the hypothesis of Lemma~\ref{le:2approx} hold
and $\leftmost()$ and $\rightmost()$ have the correct values.

\begin{figure}
	\centering
	\includegraphics[page=3, width=.9\textwidth]{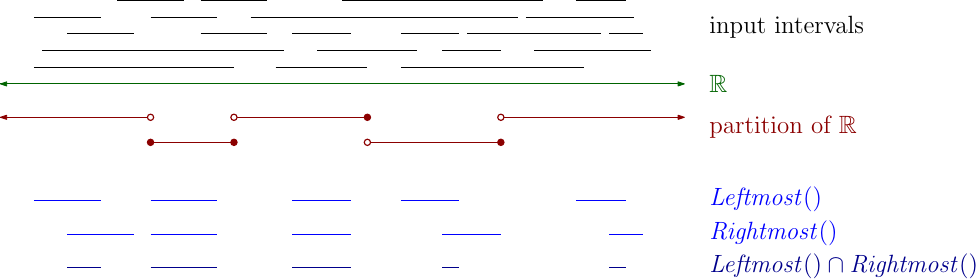}
	\caption{Data maintained by the algorithm.}
	\label{fig:partition3}
\end{figure}

With a few local operations, we can handle the insertion of a new interval in $\II$.
Consider a new interval $I$ of the stream. If $I$ is not contained in any window of
$\WW$, we do not need to do anything. If $I$ is contained in a window $W$,
we check whether it intersects all intervals contained $W$. If $I$ intersects
all intervals contained in $W$, we may have to update $\leftmost(W)$ and $\rightmost(W)$.
If $I$ does not intersect all the intervals contained in $W$, then
$I$ is disjoint from $\leftmost(W)\cap \rightmost(W)$.
In such a case we can use one endpoint of $\leftmost(W)\cap \rightmost(W)$ 
to split the window $W$ into two windows $W_1$ and
$W_2$, one containing $I$ and the other containing either 
$\leftmost(W)$ or $\rightmost(W)$, so that the 
assumptions of Lemma~\ref{le:2approx} are restored. We also have enough information
to obtain $\leftmost()$ and $\rightmost()$ for the new windows $W_1$ and $W_2$.
Figures~\ref{fig:update1} and~\ref{fig:update2} show two possible scenarios.
See the pseudocode in Figure~\ref{fig:code} for a more detailed description.

\begin{figure}
	\begin{process}{}{interval $I=[x,y]$}
		find the window $W$ of $\WW$ that contains $x$\\
		$[\ell,r]$ \qlet $\leftmost(W)\cap \rightmost(W)$\\
		\qif $y \in W$ \qthen\\
			\qif $[\ell,r]\cap [x,y]\not= \emptyset$ \qthen\\
				\qif $\ell < x$ \qor 
						$\bigl( \ell=x ~\qand~ [x,y]\subset \rightmost(W) \bigr)$ \qthen\\
					$\rightmost(W)\qlet [x,y]$ \qfi \\
				\qif $y < r$ \qor 
						$\bigl( y=r ~\qand~ [x,y]\subset \leftmost(W) \bigr)$ \qthen\\
					$\leftmost(W)\qlet [x,y]$ \qfi 
			\qelse  ~~(* $[\ell,r]$ and $[x,y]$ are disjoint; split $W$ *) \\ 
				\qif $x> r$ \qthen  ~~(* $[\ell,r]$ to the left of $[x,y]$ *)\\
					make new windows $W_1= W\cap (-\infty, r]$ and $W_2= W \cap (r, +\infty)$\\
					$\leftmost(W_1)\qlet \leftmost (W)$\\
					$\rightmost(W_1)\qlet \leftmost(W)$\\
					$I' \qlet \leftmost(W)$\\
					$\leftmost(W_2)\qlet [x,y]$\\
					$\rightmost(W_2)\qlet [x,y]$
				\qelse ~~(* $y< \ell$, $[\ell,r]$ to the right of $[x,y]$*)\\
					make new windows $W_1= W\cap (-\infty, \ell)$ and $W_2= W \cap [\ell, +\infty)$\\
					$\leftmost(W_1)\qlet [x,y]$\\
					$\rightmost(W_1)\qlet [x,y]$\\
					$\leftmost(W_2)\qlet \rightmost(W)$\\
					$\rightmost(W_2)\qlet \rightmost(W)$\\
					$I' \qlet \rightmost(W)$
				\qfi\\
				remove $W$ from $\WW$\\
				add $W_1$ and $W_2$ to $\WW$\\
				remove from $\JJ$ the interval that is contained in $W$\\
				add to $\JJ$ the intervals $[x,y]$ and $I'$
			\qfi
		\qfi\\
		(* If $y\notin W$ then $[x,y]$ is not contained in any window *)
	\end{process}
	\caption{Policy to process a new interval $[x,y]$. 
		$\WW$ maintains a partition of the real line
		and $\JJ$ maintains a $2$-approximation to $\alpha(\II)$.}
	\label{fig:code}
\end{figure}

\begin{figure}
	\centering
	\includegraphics[page=4]{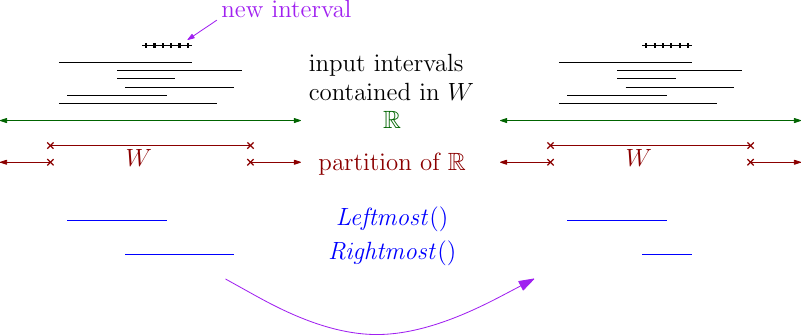}
	\caption{Example handled by lines 5--6 of the algorithm. The endpoints represented 
		by crosses may be in the window or not.}
	\label{fig:update1}
\end{figure}

\begin{figure}
	\centering
	\includegraphics[page=5]{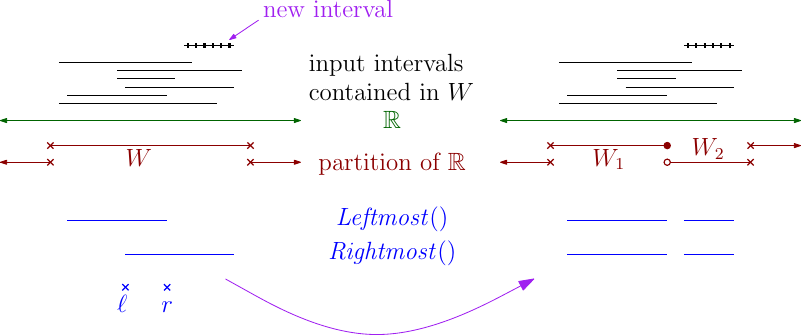}
	\caption{Example handled by lines 11--16 of the algorithm. The endpoints represented 
		by crosses may be in the window or not.}
	\label{fig:update2}
\end{figure}

\begin{lemma}
\label{le:policy}
	The policy described in Figure~\ref{fig:code} maintains a set of windows $\WW$
	and a set of intervals $\JJ$ that satisfy the assumptions of Lemma~\ref{le:2approx}.
\end{lemma}
\begin{proof}
	A simple case analysis shows that the policy maintains the assumptions
	of Lemma~\ref{le:2approx} and the properties of $\leftmost()$ and $\rightmost()$.
	
	Consider, for example, the case when the new interval $[x,y]$ 
	is contained in a window $W\in \WW$
	and $[\ell,r]=\leftmost(W)\cap \rightmost(W)$ is to the left of $[x,y]$. 
	In this case, the algorithm will update the structures in lines 11--16
	and lines 24--27. 
    See Figure~\ref{fig:update2} for an example.
	By inductive hypothesis, all the intervals in $\II\setminus \{[x,y] \}$ 
	contained in $W$ intersect $[\ell,r]$. Note that $W_1=W\cap (-\infty,r]$,
	and thus only the intervals contained in $W$ with right endpoint $r$ are contained in $W_1$.
	By the inductive hypothesis, $\leftmost(W)$ has right endpoint $r$
	and has largest left endpoint among all intervals contained in $W_1$.
	Thus, when we set
	$\rightmost(W_1)=\leftmost(W_1)=\leftmost(W)$, the correct values for $W_1$ are set.
	As for $W_2=W\cap (r,+\infty)$,
	no interval of $\II\setminus \{[x,y] \}$ is contained in $W_2$, thus 
	$[x,y]$ is the only interval contained in $W_2$ and setting
	$\rightmost(W_2)=\leftmost(W_2)=[x,y]$ we get the correct values for $W_2$.
	Lines 24--27 take care to replace $W$ in $\WW$ by $W_1$ and $W_2$.
	For $W_1$ and $W_2$ we set the correct values of 
	$\leftmost()$ and $\rightmost()$ and the assumptions of Lemma~\ref{le:2approx} hold.
	For the other windows of $\WW\setminus \{W \}$ nothing is changed.\qedhere
\end{proof}

We can store the partition of the real line $\WW$ using a dynamic binary search
tree. With this, line 1 and lines 24--25 take $O(\log |\WW|)=O(\log \alpha(\II))$ time.
The remaining steps take constant time. 
The space required by the data structure is $O(|\WW|)=O(\alpha(\II))$.
This shows the following result.

\begin{theorem}
\label{thm:lis}
	Let $\II$ be a set of intervals in the real line that arrive in a data stream.
	There is a data stream algorithm to compute a $2$-approximation to the largest 
	independent subset of $\II$ that uses $O(\alpha(\II))$ space and handles each
	interval of the stream in $O(\log \alpha(\II))$ time.
\end{theorem}

\section{Size of largest independent set of intervals}
\label{sec:slis}

In this section we show how to obtain a randomized estimate
of the value $\alpha(\II)$. 
We will assume that the endpoints of the intervals are in $[n]$. 

Using the approach of Knuth~\cite{k-75},
the algorithm presented in Section~\ref{sec:lis} can be used to define an estimator
whose expected value lies between $\alpha(\II)/2$ and $\alpha(\II)$.
However, it has large variance and we cannot use it to obtain an estimate of $\alpha(\II)$
with good guarantees. The precise idea is as follows.

The windows appearing through the algorithm of Section~\ref{sec:lis} 
naturally define a rooted binary tree $T$, where each node represents a window. 
At the root of $T$ we have the whole real line. 
Whenever a window $W$ is split into two windows $W'$ and $W''$, 
in $T$ we have nodes for $W'$ and $W''$ with parent $W$. 
The size of the output is the number of windows in the final partition,
which is exactly the number of leaves in $T$. 
Knuth~\cite{k-75} shows how to obtain an unbiased estimator of the number of leaves
of a tree. This estimator is obtained by choosing random root-to-leaf paths.
(At each node, one can use different rules to select how the random path continues.)
Unfortunately, the estimator has very large variance and cannot be used to obtain
good guarantees. Easy modifications of the method do not seem to work, 
so we develop a different method.

Our idea is to carefully split the window $[1,n]$ into segments,
and compute for each segment a 2-approximation. If each segment contains 
enough disjoint intervals from the input, then we do not do much error combining the
results of the segments.
We then have to estimate the number of segments in the partition of $[1,n]$
and the number of independent intervals in each segment.
First we describe the ingredients, independent of the streaming model,
and discuss their properties.
Then we discuss how those ingredients can be computed in the data streaming model.

\subsection{Segments and their associated information}
Let $T$ be a balanced segment tree on the $n$ segments $[i,i+1)$, $i\in [n]$.
Each leaf of $T$ corresponds to a segment $[i,i+1)$ and the order of the leaves in $T$
agrees with the order of their corresponding intervals along the real line.
Each node $v$ of $T$ has an associated segment, denoted by $S(v)$, 
that is the union of all segments stored at its descendants.
It is easy to see that, for any interval node $v$ with children $v_\ell$ and $v_r$,
the segment $S(v)$ is the disjoint union of $S(v_\ell)$ and $S(v_r)$. 
See Figure~\ref{fig:segmenttree} for an example.
We denote the root of $T$ by $r$. We have $S(r)=[1,n+1)$.

\begin{figure}
	\centering
	\includegraphics[page=1]{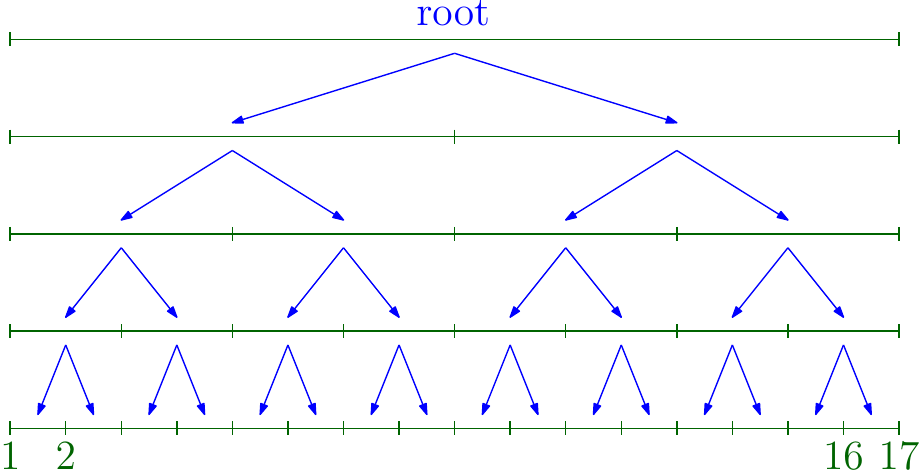}
	\caption{Segment tree for $n=16$.}
	\label{fig:segmenttree}
\end{figure}

Let $\SS$ be the set of segments associated with all nodes of $T$. 
Note that $\SS$ has $2n-1$ elements.
Each segment $S\in \SS$ contains the left endpoint and does not contain the right endpoint.

For any segment $S\in\SS$, where $S\not= S(r)$, let $\pi(S)$ be the ``parent" segment
of $S$: this is the segment stored at the parent of $v$, where $S(v)=S$.

For any $S\in \SS$, let $\beta(S)$ be the size of the largest
independent subset of $\{ I\in \II \mid I\subset S\}$. That is, we consider
the restriction of the problem to intervals of $\II$ contained in $S$.
Similarly, let $\hat\beta(S)$ be the size of a feasible solution computed
for $\{ I\in \II \mid I\subset S\}$ by the 2-approximation algorithm described
in Section~\ref{sec:lis} or by the algorithm
of Emek, Halld{\'o}rsson and Ros{\'e}n~\cite{EmekHR12}.
We thus have $\beta(S)\ge \hat\beta(S)\ge \beta(S)/2$ for all $S\in\SS$.

\begin{lemma}
\label{le:addingtree}
	Let $\SS'\subset \SS$ be a set of segments with the following properties:
	\begin{itemize}
		\item[\textup{(i)}] $S(r)$ is the disjoint union of the segments in $\SS'$, and,
		\item[\textup{(ii)}] for each $S\in \SS'$, 
			we have $\beta(\pi(S))\ge 2 \eps^{-1} \lceil \log n\rceil$.
	\end{itemize}
	Then,
	\[
		\alpha (\II) ~\ge~
		\sum_{S\in \SS'} \hat\beta(S) ~\ge~ \left(\frac{1}{2}-\eps\right)\, \alpha (\II).
	\]
\end{lemma}
\begin{proof}
	Since the segments in $\SS'$ are disjoint because of hypothesis (i), 
	we can merge the solutions giving $\beta(S)$ independent intervals, 
	for all $S\in \SS'$, to obtain a feasible solution for the whole $\II$. 
	We conclude that
	\[
		\alpha (\II) ~\ge~ \sum_{S\in \SS'} \beta(S) ~\ge~ \sum_{S\in \SS'} \hat\beta (S).
	\]
	This shows the first inequality. 
	
	Let $\tilde \SS$ be the set of leafmost elements in 
	the set of parents $\{ \pi(S)\mid S\in \SS'\}$.
    Thus, each $\tilde S\in \tilde\SS$ has some child in $\SS'$ and 
    no descendant in $\tilde \SS$.
	For each $\tilde S\in \tilde \SS$, let $\Pi_T(\tilde S)$ be the path
	in $T$ from the root to $\tilde S$.
	By construction, for each $S\in \SS'$ there exists some $\tilde S\in \tilde \SS$
	such that the parent of $S$ is on $\Pi_T(\tilde S)$.
	By assumption (ii), for each $\tilde S\in \tilde \SS$, 
	we have $\beta(\tilde S)\ge 2 \eps^{-1}\lceil \log n\rceil$.
	Each $\tilde S\in \tilde\SS$ is going to ``pay" for the error we make in the sum
	at the segments whose parents belong to $\Pi_T(\tilde S)$.

	Let $\JJ^*\subset \II$ be an optimal solution to the interval selection problem. 
	For each segment $S\in \SS$, $\JJ^*$ has at most 2 intervals that 
	intersect $S$ but are not contained in $S$. 
	Therefore, for all $S\in \SS$ we have that
	\begin{equation} \label{eq0}
		|\{ J\in \JJ^*\mid J\cap S\not= \emptyset \}| ~\le~
		|\{ J\in \JJ^*\mid J\subset S \}| +2 ~\le~ 
		\beta(S)+2.
	\end{equation}
	
	The segments in $\tilde \SS$ are pairwise disjoint because in $T$ none is a descendant
	of the other. This means that we can join solutions obtained inside the segments
	of $\tilde \SS$ into a feasible solution. Combining this with hypothesis~(ii) we get
    \begin{equation}
    	\label{eq1}
       	|\JJ^*| ~\ge~ \sum_{\tilde S\in \tilde\SS} \beta(\tilde S)
		~\ge~ |\tilde\SS|\cdot 2 \eps^{-1}\lceil \log n\rceil.
    \end{equation}

	For each $\tilde S\in \tilde \SS$, the path $\Pi_T(\tilde S)$
	has at most $\lceil\log n\rceil$ vertices.
	Since each $S\in\SS'$ has a parent in $\Pi_T(\tilde S)$, for some $\tilde S\in \tilde \SS$,
	we obtain from equation~\eqref{eq1} that
	\begin{equation} \label{eq2}
		|\SS' | ~\le~ 2\lceil\log n\rceil \cdot |\tilde \SS| 
		~\le~ 2\lceil\log n\rceil \cdot \frac{|\JJ^*|}{2 \eps^{-1}\lceil\log n\rceil} 
		~=~ \eps \cdot |\JJ^*|.
	\end{equation}
	
	Using that $S(r)$ is the union of the segments in $\SS'$ and equation~\eqref{eq0} we obtain
	\begin{align*}
		|\JJ^*| ~&\le~ \sum_{S\in \SS'} |\{ J\in \JJ^*\mid J\cap S\not= \emptyset \}|
		\\ & \le~  \sum_{S\in \SS'} (\beta(S) +2)
		\\ & =~  2\cdot |\SS'| + \sum_{S\in \SS'} \beta(S)
		\\ & \le~  2\eps \cdot |\JJ^*| + \sum_{S\in \SS'} \beta(S),
	\end{align*}
	where in the last inequality we used equation~\eqref{eq2}. 
	Now we use that
	\begin{equation*} 
		\forall S\in \SS:~~~~
			2\cdot \hat\beta(S) ~\ge~ \beta(S)
	\end{equation*}
	to conclude that
	\begin{align*}
		|\JJ^*| ~\le~  2\eps \cdot |\JJ^*| + \sum_{S\in \SS'} \beta(S) 
				~\le~  2\eps \cdot |\JJ^*| + \sum_{S\in \SS'} 2\cdot \hat\beta(S).
	\end{align*}
	The second inequality that we want to show follows because $|\JJ^*|=\alpha(\II)$.\qedhere
\end{proof}

We would like to find a set $\SS'$ satisfying the hypothesis of Lemma~\ref{le:addingtree}.
However, the definition should be local: to know whether a segment $S$ belongs to $\SS'$
we should use only local information around $S$.
The estimator $\hat\beta(S)$ is not suitable.
For example, it may happen that, for some segment $S\in \SS\setminus \{S(r)\}$,
we have $\hat\beta(\pi(S))< \hat\beta(S)$, which is counterintuitive and problematic.
We introduce another estimate that is an $O(\log n)$-approximation but is monotone nondecreasing
along paths to the root.

For each segment $S\in \SS$ we define
\[
	\gamma(S) ~=~ | \{ S'\in \SS\mid S'\subset S \text{ and } 
						\exists I\in \II \text{ s.t. } I\subset S'\} | . 
\]
Thus, $\gamma(S)$ is the number of segments of $\SS$ that are contained in $S$ and contain
some input interval.
\begin{lemma}
\label{le:gamma}
	For all $S\in \SS$, we have the following properties:
	\begin{itemize}
		\item[\textup{(i)}] $\gamma(S)\le \gamma(\pi(S))$, if $S\neq S(r)$,
		\item[\textup{(ii)}] $\gamma(S)\le \beta(S)\cdot \lceil\log n\rceil$,
		\item[\textup{(iii)}] $\gamma(S)\ge \beta(S)$, and
		\item[\textup{(iv)}] $\gamma(S)$ can be computed in $O(\gamma(S))$ space
			using the portion of the stream after the first interval contained in $S$.
	\end{itemize}
\end{lemma}
\begin{proof}
	Property (i) is obvious from the definition because any $S'$ contained in $S$
	is also contained in the parent $\pi(S)$.
	
	For the rest of the proof, fix some $S\in \SS$ and define  
	\[
		\SS' ~=~ \{ S'\in \SS\mid S'\subset S \text{ and } 
						\exists I\in \II \text{ s.t. } I\subset S'\}.
	\]
    Note that $\gamma(S)$ is the size of $\SS'$. 
    Let $T_S$ the subtree of $T$ rooted at $S$.
    
    For property (ii), note that
	$T_S$ has at most $\lceil\log n\rceil$ levels. 
	By the pigeonhole principle, there is some level $L$ of $T_S$ 
	that contains at least $\gamma(S)/ \lceil\log n\rceil$ different intervals of $\SS'$.
	The segments of $\SS'$ contained in level $L$ are disjoint, and each of them
	contains some intervals of $\II$. Picking an interval from each $S'\in L$,
	we get a subset of intervals from $\II$ that are pairwise disjoint, and
	thus $\beta(S)\ge \gamma(S)/ \lceil\log n\rceil$. 
	
	For property (iii), consider an optimal solution $\JJ^*$ for the interval selection
	problem in $\{ I\in \II \mid I\subset S\}$. Thus $|\JJ^*|=\beta(S)$.
	For each interval $J\in \JJ^*$, let $S(J)$ be the smallest $S\in \SS$ that contains $J$.
	Then $S(J)\in \SS'$. Note that $J$ contains the middle point of $S(J)$,
	as otherwise there would be a smaller segment in $\SS$ containing $J$.
	This implies that the segments $S(J)$, $J\in \JJ^*$, are all distinct.
    (However, they are not necessarily disjoint.)
	We then have
	\[
		\gamma(S) ~=~ |\SS'| ~\ge~ |\{ S(J) \mid J\in \JJ^* \}| ~=~ |\JJ^*| ~=~ \beta(S).
	\]
	
	For property (iv), we store the elements of $\SS'$ in a binary search tree.
	Whenever we obtain an interval $I$, we check whether the segments contained in $S$ and
	containing $I$ are already in the search tree and, if needed, update the structure.
	The space needed in a binary search tree is proportional to the number of elements
	stored and thus we need $O(\gamma(S))$ space.
\end{proof}

A segment $S$ of $\SS$, $S\neq S(r)$, is \DEF{relevant} if 
\[
	\gamma(\pi(S))\ge 2 \eps^{-1}\lceil\log n\rceil^2
	\text{ ~and~ }
	1~\le~ \gamma(S)< 2 \eps^{-1}\lceil\log n\rceil^2.
\]
Let $\SS_{rel} \subset \SS$ be the set of relevant segments.
If $\SS_{rel}$ is empty, then we take $\SS_{rel}=\{ S(r)\}$.

Because of Lemma~\ref{le:gamma}(i), $\gamma(\cdot)$ is nondecreasing 
along a root-to-leaf path in $T$.
Using Lemmas~\ref{le:addingtree} and~\ref{le:gamma}, we obtain the following.

\begin{lemma}
\label{le:addingtree2}
	We have
	\[
		\alpha (\II) ~\ge~ \sum_{S\in \SS_{rel}} \hat\beta(S) ~\ge~ \left(\frac{1}{2}-\eps\right) \alpha (\II).
	\]
\end{lemma}
\begin{proof}
	If $\gamma(S(r))< 2 \eps^{-1}\lceil\log n\rceil^2$, then $\SS_{rel}=\{ S(r)\}$
    and the result is clear. Thus we can assume that 
    $\gamma(S(r))\ge 2 \eps^{-1}\lceil\log n\rceil^2$, which implies
    that $S(r)\notin \SS_{rel}$.

	Define 
	\[
		\SS_0 ~=~ \{ S\in \SS\setminus \{S(r)\} \mid \gamma(S)=0 \text{ and }
					\gamma(\pi(S))\ge 2 \eps^{-1}\lceil\log n\rceil^2 \}.
	\]
	First note that the segments of $\SS_{rel}\cup \SS_0$ form a disjoint union of $S(r)$.
	Indeed, for each elementary segment $[i,i+1)\in \SS$, there exists exactly
	one ancestor that is either relevant or in $\SS_0$. 
	Lemma~\ref{le:gamma}(ii), the definition of relevant segment, and the
	fact $\gamma(S(r))\ge 2 \eps^{-1}\lceil\log n\rceil^2$ imply that
	\[
		\forall S\in \SS_{rel}\cup \SS_0:~~~~	
		\beta(\pi(S)) ~\ge~ \gamma(\pi(S))/ \lceil\log n\rceil 
		~\ge~ 2 \eps^{-1}\lceil\log n\rceil.
	\]
	Therefore, the set $\SS' = \SS_{rel}\cup \SS_0$ satisfies the conditions of Lemma~\ref{le:addingtree}.
	Using that for all $S\in \SS_0$ we have $\gamma(S)=\hat\beta(S)=0$, 
	we obtain the claimed inequalities.	
\end{proof}

Let $N_{rel}$ be the number of relevant segments.
A segment $S\in \SS$ is \DEF{active} if $S=S(r)$ or 
its \emph{parent} contains some input interval.
See Figure~\ref{fig:active} for an example.
Let $N_{act}$ be the number of active segments in $\SS$.
We are going to estimate $N_{act}$, the ratio $N_{rel}/N_{act}$, 
and the average value of $\hat\beta(S)$ over the relevant segments $S\in \SS_{rel}$.
With this, we will be able to estimate the sum considered in Lemma~\ref{le:addingtree2}.
The next section describes how the estimations are obtained in the data streaming model.

\begin{figure}
	\centering
	\includegraphics[page=2]{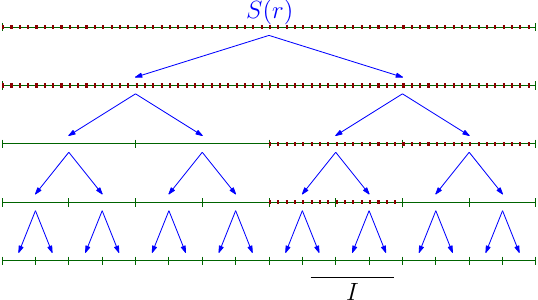}
	\caption{Active segments because of an interval $I$.}
	\label{fig:active}
\end{figure}

\subsection{Algorithms in the streaming model}

For each interval $I$, 
we use $\sigma_{\SS}(I)$ for the sequence of segments from $\SS$ 
that are active because of interval $I$, ordered \emph{non-increasingly} by size. 
Thus, $\sigma_{\SS}(I)$ contains $S(r)$ followed by
the segments whose parents contain $I$. The selected ordering implies that 
a parent $\pi(S)$ appears before $S$, for all $S$ in the sequence $\sigma_{\SS}(I)$. 
Note that $\sigma_{\SS}(I)$ has at most 
$2\lceil\log n \rceil$ elements because $T$ is balanced.

\begin{lemma}
\label{le:estimate-act}
	There is an algorithm in the data stream model that uses $O(\eps^{-2}+\log n)$ space
	and computes a value $\hat N_{act}$ such that
	\[
	\Pr \Bigl[ | N_{act}- \hat N_{act}| \le \eps \cdot N_{act} \Bigr]
		~\ge~ \frac{11}{12}.
	\] 
\end{lemma}
\begin{proof}
	We estimate $N_{act}$ using, as a black box, known results to estimate
	the number of distinct elements in a data stream.
	The stream of intervals $\II=I_1,I_2,\dots $ defines a stream of segments
	$\sigma=\sigma_{\SS}(I_1), \sigma_{\SS}(I_2),\dots $ that is $O(\log n)$ times longer.
	The segments appearing in the stream $\sigma$ are
	precisely the active segments.

	We have reduced the problem to the problem of how many distinct elements
	appear in a stream of segments from $\SS$.
	The result of Kane, Nelson and Woodruff~\cite{optimal} for distinct
	elements uses $O(\eps^{-2}+\log |\SS|)= O(\eps^{-2}+\log n)$ space and 
	computes a value $\hat N_{act}$ such that
	\[
		\Pr \left[ (1-\eps)N_{act} \le \hat N_{act} \le (1+\eps)N_{act} \right]~\ge~ \frac{11}{12}.
	\]
	Note that, to process an interval of the stream $\II$, we have to process $O(\log n)$
	segments of $\SS$. \qedhere
\end{proof}	

\begin{lemma}
\label{le:estimate-rel}
	There is an algorithm in the data stream model that uses $O(\eps^{-4}\log^4 n)$ space
	and computes a value $\hat N_{rel}$ such that
	\[
	\Pr \Bigl[ | N_{rel}- \hat N_{rel}| \le \eps \cdot N_{rel} \Bigr]
		~\ge~ \frac{10}{12}.
	\] 
\end{lemma}
\begin{proof}
	The idea is the following.
	We estimate $N_{act}$ by $\hat N_{act}$ using Lemma~\ref{le:estimate-act}.
	We take a sample of active segments, and count how many of them are relevant. 
	To get a representative sample,
	it will be important to use a lower bound on $N_{rel}/N_{act}$.
	With this we can estimate $N_{rel}= (N_{rel}/N_{act}) \cdot N_{act}$ accurately.
	We next provide the details.
	
	In $T$, each relevant segment $S'\in \SS_{rel}$ has
	$2\gamma(S')< 4 \eps^{-1}\lceil\log n\rceil^2$ active segments below it
	and at most $2\lceil\log n\rceil$ active segments whose parent is an ancestor of $S'$.
	This means that for each relevant segment there are at most 
	\[
		4 \eps^{-1}\lceil\log n\rceil^2 + 2\lceil\log n\rceil ~\le~ 6 \eps^{-1}\lceil\log n\rceil^2
	\]
	active segments.
	We obtain that
	\begin{equation}\label{eq:estimate1}
		\frac{N_{rel}}{N_{act}}~\ge~ \frac{1}{6 \eps^{-1}\lceil\log n\rceil^2}
		~=~ \frac{\eps}{6 \lceil\log n\rceil^2}.
	\end{equation}
	
	Fix any injective mapping $b$ between $\SS$ and $[n^2]$ that can be easily computed.
	For example, for each segment $S=[x,y)$ we may take $b(S)=n(x-1)+(y-1)$.
	Consider a family $\H=\H(n^2,\eps)$ of permutations $[n^2] \rightarrow [n^2]$ guaranteed
	by Lemma~\ref{le:permutations}. 
	For each $h\in \H$, the function $h\circ b$ gives an order among the elements 
	of $\SS$. We use them to compute $\H$-random samples among the active segments.
	
	Set $k=\lceil 72 \lceil\log n \rceil^2 /(\eps^3(1-\eps)]\rceil=\Theta(\eps^{-3}\log^2 n)$, and choose
	permutations $h_1,\dots,h_k\in \H$ uniformly and independently at random.
	For each permutation $h_j$, where $j=1,\dots, k$, let $S_j$ be the 
	\emph{active} segment of $\SS$ that minimizes $(h_j\circ b)(\cdot)$. Thus
	\[
		S_j ~=~ \arg \min \Bigl\{ h_j(S)\mid S\in \SS 
			\text{ is active} \Bigr\}.
	\]
	The idea is that $S_j$ is nearly a random active segment of $\SS$.
	Therefore, if we define the random variable
	\[
		X= \bigl|\{ j\in \{1,\dots, k\}\mid S_j\text{ is relevant}\}\bigr|
	\]
	then $N_{rel}/N_{act}$ is roughly $X/k$. Below we discuss the computation of $X$.		
	To analyze the random variable $X$ more precisely,
	let us define
	\[
		p ~=~ \Pr_{h_j\in\H}[S_j \text{ is relevant} ].
	\]
	Since $S_j$ is selected among the active segments,
	the discussion after Lemma~\ref{le:permutations} implies
	\begin{equation}
		\label{eq:estimate2}
		p ~\in~ 
		\left[ \frac{(1-\eps) N_{rel}}{N_{act}}, \frac{(1+\eps) N_{rel}}{N_{act}}\right].
	\end{equation}
	In particular, using the estimate \eqref{eq:estimate1} and the definition of $k$ we get
	\begin{equation}
		\label{eq:estimate3}
			kp ~\ge~ 
			\frac{72 \lceil\log n \rceil^2}{\eps^3(1-\eps)}\cdot \frac{(1-\eps) N_{rel}}{N_{act}} 
			~\ge~ 
			\frac{72 \lceil\log n \rceil^2}{\eps^3} \cdot \frac{\eps}{6 \lceil\log n\rceil^2}
			~=~
			\frac{12}{\eps^2}.
	\end{equation}
	Note that $X$ is the sum of $k$ independent random variables 
	taking values in $\{ 0,1\}$ and $\mathbb{E}[X]=kp$.
	It follows from Chebyshev's inequality and the lower bound in~\eqref{eq:estimate3}
	that
	\begin{align*}
		\Pr\left[ \left| \frac Xk- p \right| \ge \eps p \right] ~&=~
		\Pr\Bigl[ | X- kp| \ge \eps kp \Bigr] ~~\le~ \frac{\Var[X]}{(\eps k p)^2} \\
		&=~ \frac{kp(1-p)}{(\eps kp)^2} ~<~ \frac{1}{kp \eps^2}
		~\le~ \frac{1}{12}.
	\end{align*}
	
	To finalize, let us define the estimator 
	$\hat N_{rel} = \hat N_{act} \cdot \left( \frac{X}{k} \right)$ of $N_{rel}$,
	where $\hat N_{act}$ is the estimator of $N_{act}$ given in Lemma~\ref{le:estimate-act}.
	When the events
	\[
		\left[ | N_{act}-\hat N_{act}| \le \eps N_{act} \right] 
		~~~\text{and}~~~
		\left[ \left| \frac Xk- p \right| \le \eps p \right]
	\]
	occur, then we can use equation~\eqref{eq:estimate2} and $\eps<1/2$ to see that
	\begin{align*}
		\hat N_{rel} ~&\le~ (1+\eps)N_{act}\cdot (1+\eps) p ~\le~
		(1+\eps)^2 N_{act}\cdot  \frac{(1+\eps) N_{rel}}{N_{act}}
		~=~ (1+\eps)^3 N_{rel} \\ &\le (1+7\eps)N_{rel},
	\end{align*}
	and also
	\begin{align*}
		\hat N_{rel} ~&\ge~ (1-\eps)N_{act}  \cdot (1-\eps) p ~\ge~
		(1-\eps)^2 N_{act} \cdot \frac{(1-\eps) N_{rel}}{N_{act}}
		~=~ (1-\eps)^3 N_{rel}\\
		&\ge~ (1-7\eps) N_{rel}.
	\end{align*}
	We conclude that
	\begin{align*}
		\Pr\Bigl[ (1-7\eps) N_{rel} \le \hat N_{rel} \le (1+7\eps) N_{rel} \Bigr]
		~&\ge~ 1-\Pr \left[ | N_{act}-\hat N_{act}| \ge \eps N_{act} \right] -
			\Pr \left[ \left| \frac Xk- p \right| \ge \eps p \right] \\
		~&\ge~ 1-\frac{1}{12}-\frac{1}{12} ~\ge~ \frac{10}{12}.
	\end{align*}
	Replacing in the argument $\eps$ by $\eps/7$, we obtain the desired bound.

	It remains to discuss how $X$ can be computed.
	For each $j$, where $j=1,\dots, k$,
	we keep a variable that stores the current segment $S_j$ for
	all the segments that are active so far, 
    keep information about the choice of $h_j$,
    and keep information
	about $\gamma(S_j)$ and $\gamma(\pi(S_j))$, so that we can decide whether $S_j$
	is relevant. 

	Let $I_1,I_2,\dots $ be the data stream of input intervals. 
	We consider the stream of segments
	$\sigma=\sigma_\SS(I_1), \sigma_\SS(I_2),\dots $.
	When handling a segment $S$ of the stream $\sigma$,
	we may have to update $S_j$; this happens when $h_j(S) < h_j(S_j)$. 
	Note that we can indeed maintain $\gamma(\pi(S_j))$
	because $S_j$ becomes active the first time that its parent contains
	some input interval. This is also the first time when $\gamma(\pi(S_j))$
	becomes nonzero, and thus the forthcoming part of the stream has enough
	information to compute $\gamma(S_j)$ and $\gamma(\pi(S_j))$.	
	(Here it is convenient that $\sigma_\SS(I)$ gives segments in decreasing size.)
	To maintain $\gamma(S_j)$ and $\gamma(\pi(S_j))$, we use Lemma~\ref{le:gamma}(iv).
	
	To reduce the space used by each index $j$, we use the following simple trick.
	If at some point we detect that $\gamma(S_j)$ is larger than 
	$2 \eps^{-1}\lceil\log n\rceil^2$, we just store that $S_j$ is not relevant.
	If at some point we detect that $\gamma(\pi(S_j))$ is larger than 
	$2 \eps^{-1}\lceil\log n\rceil^2$, we just store that $\pi(S_j)$ is large enough
	that $S_j$ could be relevant. 	
	We conclude that, for each $j$, we need at most $O(\log(1/\eps)+\eps^{-1}\log^2 n)$ space.
	Therefore, we need in total
	$O(k \eps^{-1}\log^2 n)=O(\eps^{-4}\log^4 n)$ space.
\qedhere
\end{proof}

Let 
\[
	\rho ~=~ \left(\sum_{S\in \SS_{rel}} \hat\beta(S)\right) / |\SS_{rel}|.
\]
The next result shows how to estimate $\rho$.

\begin{lemma}
\label{le:estimate-contribution}
	There is an algorithm in the data stream model that uses $O(\eps^{-5}\log^6 n)$ space
	and computes a value $\hat \rho$ such that
	\[
	\Pr \Bigl[ | \rho- \hat \rho| \le \eps \rho \Bigr]
		~\ge~ \frac{10}{12}.
	\] 
\end{lemma}
\begin{proof}
	Fix any injective mapping $b$ between $\SS$ and $[n^2]$ that can be easily computed,
	and consider a family $\H=\H(n^2,\eps)$ of permutations $[n^2] \rightarrow [n^2]$ guaranteed
	by Lemma~\ref{le:permutations}. 
	For each $h\in \H$, the function $h\circ b$ gives an order among the elements 
	of $\SS$. We use them to compute $\H$-random samples among the active segments.

	Let $\SS_{act}$ be the set of active segments.
    
	Consider a random variable $Y_1$ defined as follows. 
    We repeatedly sample $h_1\in \H$ uniformly at random,
    until we get that $\arg\min_{S\in \SS_{act}} \hat\beta(S)$ is a relevant segment.
	Let $S_1$ be the resulting relevant segment $\arg\min_{S\in \SS_{act}} \hat\beta(S)$,
	and set $Y_1=\hat\beta(S_1)$.
    Because of Lemma~\ref{le:conditional}, where $X=\SS_{act}$ and $Y=\SS_{rel}$, we have
    \[
    	\forall S\in \SS_{rel}:~~~~ \frac{1-4\eps}{|\SS_{rel}|} ~\le~ 
        			\Pr [S_1=S]~\le~ \frac{1+4\eps}{|\SS_{rel}|}.
    \]	
	We thus have
	\[
		\mathbb{E}[Y_1] ~=~ \sum_{S\in \SS_{rel}} \Pr[S_1=S]\cdot \hat\beta(S) ~\le\
		\sum_{S\in \SS_{rel}} \frac{1+4\eps}{|\SS_{rel}|}\cdot \hat\beta(S) ~=~
		(1+4\eps) \cdot \rho,
	\]
	and similarly
	\[
		\mathbb{E}[Y_1] ~\ge~
		\sum_{S\in \SS_{rel}} \frac{1-4\eps}{|\SS_{rel}|}\cdot \hat\beta(S) ~=~
		(1-4\eps) \cdot \rho.
	\]
	For the variance we can use $\hat \beta(S)\le \gamma(S)$ and 
    the definition of relevant segments to get
	\begin{align*}
		\Var[Y_1] ~ &\le ~ \mathbb{E}[Y_1^2] \\
        &=~ 
		\sum_{S\in \SS_{rel}} \Pr[S_1=S]\cdot (\hat\beta(S))^2 \\
		&\le~ 
		\sum_{S\in \SS_{rel}} \frac{1+4\eps}{|\SS_{rel}|}\cdot \hat\beta(S) \cdot  
				\frac{2 \lceil\log n\rceil^2}{\eps} \\
		&\le~
		\frac{2(1+4\eps) \lceil\log n\rceil^2}{\eps} \cdot \rho \\
		&\le~ \frac{6 \lceil\log n\rceil^2}{\eps} \cdot \rho
	\end{align*}
	Note also that $\gamma(S)\ge 1$ implies $\hat\beta(S)\ge 1$.
	Therefore, $\rho\ge 1$.

	Consider an integer $k$ to be chosen later.
	Let $Y_2,\dots, Y_k$ be independent random variables with the same distribution that $Y_1$,
	and define $\hat\rho=(\sum^k_{i=1} Y_i)/k$.
	Using Chebyshev's inequality and $\rho\ge 1$ we obtain
	\begin{align*}
		\Pr\left[ \left| \hat\rho - \mathbb{E}[Y_1] \right| \ge \eps \rho \right] ~&=~
		\Pr\Bigl[ | \hat\rho k - \mathbb{E}[Y_1]k| \ge \eps k\rho \Bigr] 
		~\le~ \frac{\Var[\hat\rho k]}{(\eps k \rho)^2} \\
		&=~ \frac{k\Var[Y_1]}{(\eps k\rho)^2} 
		~<~ \frac{\frac{6 \lceil\log n\rceil^2}{\eps} \cdot \rho}{k \eps^2 \rho^2} \\
		&\le~ \frac{6 \lceil\log n\rceil^2}{k\eps^3}
	\end{align*}
	Setting $k=6\cdot 12\cdot\lceil\log n\rceil^2/\eps^3$,
	we have 
	\begin{align*}
		\Pr\left[ \left| \hat\rho - \mathbb{E}[Y_1] \right| \ge \eps \rho \right] 
		~\le~ \frac{1}{12}.
	\end{align*}

	We then proceed similar to the proof of Lemma~\ref{le:estimate-rel}.
	Set $k_0= 12 \lceil\log n\rceil^2 k/\eps(1-\eps) = \Theta(\eps^{-4}\log^4 n)$.
	For each $j\in [k_0]$, take a function $h_j\in \H$ uniformly at random
	and select $S_j=\arg\min \{ h(b(S))\mid S \text{ is active}  \}$.
	Let $X$ be the number of relevant segments in $S_1,\dots, S_{k_0}$
	and let $p=\Pr[S_1 \in \SS_{rel}]$.
	Using the analysis of Lemma~\ref{le:estimate-rel} we have
	\begin{equation*}
			k_0 p ~\ge~ 
			\frac{(12 \lceil\log n \rceil^2) k}{\eps(1-\eps)}\cdot \frac{(1-\eps)N_{rel}}{N_{act}} 
			~\ge~ 
			k\cdot \frac{12 \lceil\log n \rceil^2}{\eps(1-\eps)} \cdot \frac{(1-\eps)\eps}{6 \lceil\log n\rceil^2}
			~=~
			2k
	\end{equation*}
	and
	\[
		\Pr\Bigl[ | X- k_0 p| \ge k_0 p/2 \Bigr] ~\le~ \frac{\Var[X]}{(k_0 p/2)^2}
		~=~ \frac{4 k_0 p(1-p)}{k_0^2p^2}
		~<~ \frac{4}{k_0 p}
		~\le~ \frac{4}{2k}
		~\le~ \frac{1}{12}.
	\]	
	This means that, with probability at least $11/12$, the sample
	$S_1,\dots, S_{k_0}$ contains at least $(1/2)k_0p \ge k$ relevant segments.
	We can then use the first $k$ of those relevant segments to compute the estimate $\hat\rho$.
	
	With probability at least $1-1/12-1/12=10/12$ we have the events
	\[
		\Bigl[ | X- k_0 p| \ge k_0 p/2 \Bigr]
		~~~\text{ and }~~~
		\left[ \left| \hat\rho - \mathbb{E}[Y_1] \right| \ge \eps \rho \right].
	\] 	
	In such a case
	\[
		\hat\rho ~\le~ \eps\rho + \mathbb{E}[Y_1] ~\le~ \eps\rho + (1+4\eps)\rho ~=~ (1+5\eps)\rho
	\]
	and similarly
	\[
		\hat\rho ~\ge~ \mathbb{E}[Y_1] -\eps\rho ~\ge~ (1-4\eps)\rho - \eps\rho ~=~ (1-5\eps)\rho.
	\]
	Therefore,
	\[
	\Pr \Bigl[ | \rho- \hat \rho| \le 5\eps \rho \Bigr]
		~\ge~ \frac{10}{12}.
	\]
	Changing the role of $\eps$ and $\eps/5$, the claimed probability is obtained.
	
	It remains to show that we can compute $\hat\rho$ in the data stream model.
	Like before, for each $j\in [k_0]$, we have to maintain	the segment $S_j$, 
    information about the choice of the permutation $h_j$,
    information about $\gamma(S_j)$ and $\gamma(\pi(S_j))$,
	and the value $\hat\beta(S_j)$. Since $\hat\beta(S_j)\le \beta(S_j)\le \gamma(S_j)$
	because of Lemma~\ref{le:gamma}(iii), we need 
	$O(\eps^{-1}\log^2 n)$ space per index $j$. 
	In total we need $O(k_0\eps^{-1}\log^2 n)= O(\eps^{-5}\log^6 n)$ space.
	\qedhere
\end{proof}

\begin{theorem}
	Assume that $\eps\in (0,1/2)$ and 
    let $\II$ be a set of intervals with endpoints in $\{1,\dots, n\}$ 
    that arrive in a data stream.
	There is a data stream algorithm that uses 
    $O(\eps^{-5}\log^6 n)$ space and computes a value $\hat\alpha$
	such that
	\[
	\Pr \Bigl[ \tfrac{1}{2}\left(1-\eps\right)\cdot \alpha(\II) \le 
				\hat\alpha \le 
				\alpha(\II) \Bigr]
		~\ge~ \frac 23.
	\] 
\end{theorem}
\begin{proof}
	We compute the estimate $\hat N_{rel}$ of Lemma~\ref{le:estimate-rel}
	and the estimate $\hat \rho$ of Lemma~\ref{le:estimate-contribution}.
	Define the estimate $\hat \alpha_0 = \hat N_{rel}\cdot \hat \rho$.
	With probability at least $1-\tfrac{2}{12}-\tfrac{2}{12}=\tfrac{2}{3}$
    we have the events
	\[
		\Bigl[ | N_{rel}- \hat N_{rel}| \le \eps \cdot N_{rel} \Bigr]
		~~~\text{ and }~~~
		\Bigl[ | \rho- \hat \rho| \le \eps \rho \Bigr].
	\]
	When such events hold, we can use the definitions of $N_{rel}$ and $\rho$, 
	together with Lemma~\ref{le:addingtree2}, to obtain
	\[
		\hat\alpha_0 ~\le~ (1+\eps) N_{rel}\cdot (1+\eps) \rho ~=~ 
			(1+\eps)^2 \sum_{S\in \SS_{rel}} \hat\beta(S) ~\le~
			(1+\eps)^2 \alpha(\II)
	\]
	and 
	\[
		\hat\alpha_0 ~\ge~ (1-\eps) N_{rel}\cdot (1-\eps) \rho ~=~ 
			(1-\eps)^2 \sum_{S\in \SS_{rel}} \hat\beta(S) ~\ge~
			(1-\eps)^2 \left(\frac{1}{2}-\eps\right) \alpha(\II).
	\]
	Therefore, 
	\[
	\Pr \left[ (1-\eps)^2 \left(\frac{1}{2}-\eps\right)\cdot \alpha(\II) \le 
				\hat\alpha_0 \le 
				(1+\eps)^2 \cdot \alpha(\II) \right]
		~\ge~ \frac 23.
	\] 
	Using that $(1-\eps)^2 (1/2-\eps)/(1+\eps)^2 \ge 1/2-3\eps$ for all $\eps\in(0,1/2)$,
	rescaling $\eps$ by $1/6$, and setting $\hat\alpha=\hat\alpha_0/(1+\eps)^2$,
    the claimed approximation is obtained. 
	The space bounds are those from Lemmas~\ref{le:estimate-rel} and~\ref{le:estimate-contribution}.
\end{proof}

\section{Largest independent set of same-size intervals}
\label{sec:lisu}
In this section we show how to obtain a $(3/2)$-approximation to the largest independent
set using $O(\alpha(\II))$ space in the special case when all the intervals have the same
length $\lambda>0$. 

Our approach is based on using the shifting technique of Hochbaum and Mass~\cite{hm-85} 
with a grid of length $3\lambda$ and shifts of length $\lambda$. We observe that we can maintain an optimal solution
restricted to a window of length $3\lambda$ because at most two disjoint 
intervals of length $\lambda$ can fit in.

For any real value $\ell$, let $W_\ell$ denote the window $[\ell,\ell+3\lambda)$. 
Note that $W_\ell$ includes the left endpoint but excludes the right endpoint.
For $a\in \{ 0,1,2 \}$, we define the partition of the real line
\[
	\WW_a ~=~ \Bigl\{ W_{(a+3j)\lambda} \mid j\in \ZZ \Bigr\}.
\]
For $a\in \{ 0,1,2 \}$, let $\II_a$ be the set of input intervals contained in some window of $\WW_a$. Thus,	
\[
	\II_a ~=~ \Bigl\{ I\in \II \mid \not\exists j\in \ZZ \text{ s.t. } (a+3j)\lambda \in I \Bigr\}.
\]

\begin{lemma}
\label{le:32}
	If all the intervals of $\II$ have length $\lambda>0$, then
	\[
		\max \Bigl\{ \alpha(\II_0), \alpha(\II_1), \alpha(\II_2) \Bigr\} ~\ge~ \frac 23 \alpha(\II).
	\]
\end{lemma}
\begin{proof}
	Each interval of length $\lambda$ is contained in exactly two windows of 
	$\WW_0\cup \WW_1\cup \WW_2$. Let $\JJ^* \subseteq \II$ be a largest independent set
	of intervals, so that $|\JJ^*| =\alpha (\II)$. We then have
	\[
    	3 \cdot \max \Bigl\{ \alpha(\II_0), \alpha(\II_1), \alpha(\II_2) \Bigr\} ~\ge~
		\sum_{0\le a \le 2} \alpha(\II_a) ~\ge~ 
		\sum_{0\le a \le 2} |\JJ^*\cap \II_a| ~\ge~
		2 |\JJ^*| ~=~
		2 \alpha(\II)
	\]
	and the result follows.\qedhere
\end{proof}

For each $a\in \{0,1,2\}$ we store an optimal solution $\JJ_a$ restricted to $\II_a$.
We obtain a $(3/2)$-approximation by returning the largest among $\JJ_0, \JJ_1,\JJ_2$.

For each window $W$ considered through the algorithm,
we store $\leftmost(W)$ and $\rightmost(W)$. 
We also store a boolean value $\Active(W)$ telling whether some previous interval was contained
in $W$. When $\Active(W)$ is false, $\leftmost(W)$ and $\rightmost(W)$ are undefined.

With a few local operations, we can handle the insertion of new intervals in $\II$.
For a window $W\in \WW_a$, there are two relevant moments when $\JJ_a$ may be changed. 
First, when $W$ gets the first interval,
the interval has to be added to $\JJ_a$ and
$\Active(W)$ is set to true. 
Second, when $W$ can first fit two disjoint intervals,
then those two intervals have to be added to $\JJ_a$.  
See the pseudocode in Figure~\ref{fig:code2} for a more detailed description.

\begin{figure}
	\begin{process}{}{interval $[x,y]$ of length $\lambda$}
		\qfor $a=0,1,2$ \qdo\\
			$W \qlet$ window of $\WW_a$ that contains $x$ \\
			\qif $y \in W$ \qthen ~~~(* $[x,y]$	is contained in the window $W\in \WW_a$ *) \\
				\qif $\Active(W)$ false \qthen\\
					$active(W) \qlet$ true\\
					$\rightmost(W)\qlet [x,y]$ \\
					$\leftmost(W)\qlet [x,y]$ \\
					add $[x,y]$ to $\JJ_a$
				\qelse\qif $\rightmost(W)\cap \leftmost(W) \neq \emptyset$ \qthen\\
					$[\ell,r] \qlet \rightmost(W)\cap \leftmost(W)$\\
					\qif $\ell < x$ \qthen $\rightmost(W)\qlet [x,y]$ \qfi \\
					\qif $y < r$ \qthen $\leftmost(W)\qlet [x,y]$ \qfi \\
					\qif $\rightmost(W)\cap \leftmost(W) = \emptyset$ \qthen\\
						remove from $\JJ_a$ the interval contained in $W$\\
						add to $\JJ_a$ intervals $\rightmost(W)$ and $\leftmost(W)$\qfi
			\qfi
		\qfi
	\end{process}
	\caption{Policy to process a new interval $[x,x+\lambda]$. 
		$\JJ_a$ maintains an optimal solution for $\alpha(\II_a)$.}
	\label{fig:code2}
\end{figure}

\begin{lemma}
	For $a=0,1,2$, 
	the policy described in Figure~\ref{fig:code2} maintains an optimal solution
	$\JJ_a$ for the intervals $\II_a$.
\end{lemma}
\begin{proof}
	Since a window $W$ of length $3$ can contain at most 2 disjoint intervals
	of length $\lambda$, the intervals 
	$\rightmost(W)$ and $\leftmost(W)$ suffice to obtain an optimal solution
	restricted to intervals contained in $W$. By the definition of $\II_a$,
	an optimal solution for 
    \[
    	\bigcup_{W\in \WW_a} \{ \rightmost(W), \leftmost(W)\}
    \]
	is an optimal solution for $\II_a$. Since the algorithm maintains such
	an optimal solution, the claim follows.\qedhere
\end{proof}

Since each window can have at most two disjoint intervals and each interval
is contained in at most two windows of $\WW_0\cup \WW_1\cup \WW_2$,
we have at most $O(\alpha(\II))$ active intervals through the entire stream.
Using a dynamic binary search tree for the active windows, we can perform the operations
in $O(\log \alpha(\II))$ time. 
We summarize.

\begin{theorem}
	Let $\II$ be a set of intervals of length $\lambda$ 
	in the real line that arrive in a data stream.
	There is a data stream algorithm to compute a $(3/2)$-approximation to the largest 
	independent subset of $\II$ 
	that uses $O(\alpha(\II))$ space and handles each
	interval of the stream in $O(\log \alpha(\II))$ time.
\end{theorem}

\section{Size of largest independent set for same-size intervals}
\label{sec:slisu}

In this section we show how to obtain a randomized estimate
of the value $\alpha(\II)$ in the special case when all the intervals have the same length $\lambda>0$. We assume that the endpoints are in $[n]$.

The idea is an extension of the idea used in Section~\ref{sec:lisu}.
For $a=0,1,2$, let $\WW_a$ and $\II_a$ be as defined in 
Section~\ref{sec:lisu}. For $a=0,1,2$, we will compute
a value $\hat\alpha_a$ that $(1+\eps)$-approximates $\alpha(\II_a)$ 
with reasonable probability.
We then return $\hat \alpha=\max \{ \hat\alpha_0, \hat\alpha_1, \hat\alpha_2 \}$,
which catches a fraction at least $\tfrac{2}{3}(1-\eps)$ of $\alpha(\II)$,
with reasonable probability.

To obtain the $(1+\eps)$-approximation to $\alpha(\II_a)$,
we want to estimate how many windows of $\WW_a$ contain some input interval
and how many contain two disjoint input intervals. For
this we combine known results for the distinct elements as a black box
and use sampling over the windows that contain some input interval. 

\begin{lemma}\label{le:alpha}
	Let $a$ be $0$, $1$ or $2$ and let $\eps\in (0,1)$.
	There is an algorithm in the data stream model that uses 
	$O(\eps^{-2}\log(1/\eps)+\log n)$ space
	and computes a value $\hat\alpha_a$ such that
	\[
	\Pr \Bigl[ |\alpha(\II_a)- \hat\alpha_a| \le \eps \cdot \alpha(\II_a) \Bigr]
		~\ge~ \frac{8}{9}.
	\] 
\end{lemma}
\begin{proof}
	Let us fix some $a\in \{0,1,2 \}$.
	We say that a window $W$ of $\WW_a$
	is of type $i$ if $W$ contains \emph{at least} $i$ disjoint input intervals. 
	Since the windows of $\WW_a$ have length $3\lambda$, they
	can be of type $0$, $1$ or $2$. 
	For $i=0,1,2$, let $\gamma_i$ be the number of windows of type $i$ in $\WW_a$.
	Then $\alpha(\II_a)= \gamma_1+\gamma_2$.
	
	We compute an estimate $\hat \gamma_1$ to $\gamma_1$ as follows.
	We have to estimate $| \{ W\in \WW_a\mid \exists I \text{ s.t. } I\subset W\}|$.
	The stream of intervals $\II=I_1,I_2,\dots $ defines the sequence of windows
	$W(\II)= W(I_1), W(I_2),\dots $, 
	where $W(I_i)$ denotes the window of $\WW_a$ that contains $I_i$;
	if $I_i$ is not contained in any window of $\WW_a$, we then skip $I_i$.
	Then $\gamma_1$ is the number of distinct elements in the sequence $W(\II)$.
	The results of Kane, Nelson and Woodruff~\cite{optimal} imply that
	using $O(\eps^{-2}+\log n)$ space we can compute a value $\hat\gamma_1$ such that
	\[
		\Pr \left[ (1-\eps)\gamma_1 \le \hat\gamma_1 \le (1+\eps)\gamma_1 \right]~\ge~ \frac{17}{18}.
	\]
	
	We next explain how to estimate the ratio $\gamma_2/\gamma_1\le 1$.
	Consider a family $\H=\H(n,\eps)$ of permutations $[n]\rightarrow [n]$ guaranteed
	by Lemma~\ref{le:permutations}, set $k=\lceil 18\eps^{-2}\rceil$, and choose
	permutations $h_1,\dots,h_k\in \H$ uniformly and independently at random.
	For each permutation $h_j$, where $j=1,\dots, k$, let $W_j$
	be the window $[\ell,\ell+3\lambda)$ of $\WW_a$ that contains some input interval
	and minimizes $h_j(\ell)$. Thus
	\[
		W_j ~=~ \arg \min \Bigl\{ h_j(\ell)\mid [\ell ,\ell +3\lambda)\in \WW_a, 
			\text{ some $I\in \II$ is contained in $[\ell ,\ell +3\lambda)$} \Bigr\}.
	\]
	The idea is that $W_j$ is a nearly-uniform random window of $\WW_a$, among
	those that contain some input interval. 
	Therefore, if we define the random variable
	\[
		M= \bigl|\{ j\in \{1,\dots, k\}\mid W_j\text{ is of type 2}\}\bigr|
	\]
	then $\gamma_2/\gamma_1$ is roughly $M/k$. Below we make a precise analysis.
	
	Let us first discuss that $M$ can be computed in 
	space within $O(k \log(1/\eps))=O(\eps^{-2}\log(1/\eps))$.
	For each $j$, where $j=1,\dots, k$,
	we keep information about the choice of $h_j$,
	keep a variable that stores the current window $W_j$ for
	all the intervals that have been seen so far,
	and store the intervals $\rightmost(W_j)$ and $\leftmost(W_j)$. 
	Those two intervals tell us whether $W_j$ is of
	type $1$ or $2$. When handling an interval $I$ of the stream,
	we may have to update $W_j$; this happens when $h_j(s) < h_j(s_j)$,
	where $s$ is the left endpoint of the window of $\WW_a$ that contains $I$
	and $s_j$ is the left endpoint of $W_j$. When $W_j$ is changed, 
	we also have to reset the intervals
	$\rightmost(W_j)$ and $\leftmost(W_j)$ to the new interval $I$.

	To analyze the random variable $M$ more precisely,
	let us define
	\[
		p ~=~ \Pr_{h_j\in\H}[W_j \text{ is of type $2$} ] ~\in~ 
		\left[ \tfrac{(1-\eps) \gamma_2}{\gamma_1}, \tfrac{(1+\eps) \gamma_2}{\gamma_1}\right].
	\]
	Note that $M$ is the sum of $k$ independent random variables 
	taking values in $\{ 0,1\}$ and $\mathbb{E}[M]=kp$.
	It follows from Chebyshev's inequality that
	\[
		\Pr\left[ \left| \frac Mk- p \right| \ge \eps \right] ~=~
		\Pr\Bigl[ | M- kp| \ge \eps k \Bigr] ~\le~
		\frac{\Var[M]}{(\eps k)^2}~\le~
		\frac{kp}{\eps^2 k^2} ~\le~
		\frac{1}{\eps^2 k} ~\le~ \frac{1}{18}.
	\]
	
	To finalize, let us define the estimator 
	$\hat \alpha_a = \hat\gamma_1 \left( 1+ \frac{M}{k} \right)$.
	When the events
	\[
		\left[ (1-\eps)\gamma_1 \le \hat\gamma_1 \le (1+\eps)\gamma_1 \right] 
		~~~\text{and}~~~
		\left[ \left| \frac Mk- p \right| \le \eps \right]
	\]
	occur, then we have
	\begin{align*}
		\hat\alpha_a ~&\le~ (1+\eps)\gamma_1  \left( 1+ p \right) ~\le~
		(1+\eps)\gamma_1  \left( 1+\eps+ \frac{(1+\eps) \gamma_2}{\gamma_1} \right)\\
		&=~ (1+\eps)^2 (\gamma_1+\gamma_2) ~\le~ (1+3\eps)\alpha(\II_a),
	\end{align*}
	and also
	\begin{align*}
		\hat\alpha_a ~&\ge~ (1-\eps)\gamma_1  \left( 1-\eps+ p \right) ~\ge~
		(1-\eps)\gamma_1  \left( 1-\eps + \frac{(1-\eps) \gamma_2}{\gamma_1} \right)\\
		&=~ (1-\eps)^2 (\gamma_1+\gamma_2) ~\ge~ (1-3\eps)\alpha(\II_a).
	\end{align*}
	We conclude that
	\[
		\Pr\Bigl[ (1-3\eps)\alpha(\II_a) \le \hat\alpha_a \le (1+3\eps)\alpha(\II_a) \Bigr]
			~\ge~ 1-\frac{1}{18}-\frac{1}{18} ~\ge~ \frac{8}{9}.
	\]
	Replacing in the argument $\eps$ by $\eps/3$, the result follows.\qedhere
\end{proof}

\begin{theorem}
	Assume that $\eps\in (0,1/2)$ and let $\II$ be a set of intervals of length $\lambda$ 
	with endpoints in $\{1,\dots, n\}$ that arrive in a data stream.
	There is a data stream algorithm that uses  
    $O(\eps^{-2}\log(1/\eps)+\log n)$ space
	and computes a value $\hat\alpha$
	such that
	\[
	\Pr \Bigl[ \tfrac 23 (1-\eps)\cdot \alpha(\II) \le 
				\hat\alpha \le 
				\alpha(\II) \Bigr]
		~\ge~ \frac 23.
	\] 
\end{theorem}
\begin{proof}
	For each $a=0,1,2$ we compute the estimate $\hat\alpha_a$ to $\alpha(\II_a)$
	with the algorithm described in Lemma~\ref{le:alpha}.
	We then have 
	\[
	\Pr \left[ \bigwedge_{a=0,1,2} \Bigl[ |\alpha(\II_a)- \hat\alpha_a| \le \eps \cdot \alpha(\II_a) \Bigr]\right]
		~\ge~ \frac{6}{9}.
	\] 
	When the event occurs, it follows 
	by Lemma~\ref{le:32} that
	\[
		\tfrac 23 (1-\eps)\cdot \alpha(\II) \le 
					\max \{ \hat\alpha_0, \hat\alpha_1, \hat\alpha_2 \} \le 
					(1+\eps) \alpha(\II).
	\] 
	Therefore, 
	using that $(1-\eps)/(1+\eps)\ge 1-2\eps$ for all $\eps\in(0,1/2)$,
	rescaling $\eps$ by $1/2$,
	and returning $\hat\alpha =\max \{ \hat\alpha_0, \hat\alpha_1, \hat\alpha_2 \}/(1+\eps)$, 
	the result is achieved.\qedhere
\end{proof}

\section{Lower bounds}
\label{sec:lower}

Emek, Halld{\'o}rsson and Ros{\'e}n~\cite{EmekHR12} showed that any streaming algorithm
for the interval selection problem cannot achieve an approximation ratio of $r$, for any
constant $r<2$. They also show that, for same-size intervals, one cannot
obtain an approximation ratio below $3/2$. We are going to show that similar
inapproximability results hold for estimating $\alpha(\II)$.

The lower bound of 
Emek et al.~uses the coordinates of the endpoints to recover 
information about a permutation.
That is, given a solution to the interval selection problem, they can use the
endpoints of the intervals to recover information. Thus, such reduction cannot
be adapted to the estimation of $\alpha(\II)$, since we do not require to 
return intervals. Nevertheless, there are certain similarities between their construction
and ours, especially for same-length intervals.

Consider the problem \Index. The input to \Index\ is a pair $(S,i)\in \{0,1\}^n\times [n]$
and the output, denoted by \Index$(S,i)$,
is the $i$-th bit of $S$. One can think of $S$ as a subset of $[n]$, and
then \Index$(S,i)$ is asking whether element $i$ is in the subset or not.

The one-way communication complexity of \Index\ is well understood. In this
scenario, Alice has $S$ and Bob has $i$. Alice sends a message to Bob and then Bob
has to compute \Index$(S,i)$. The key question is how long should be the message in
the worst case so
that Bob can compute \Index$(S,i)$ correctly with probability greater than $1/2$, say, at least $2/3$. 
(Attaining probability $1/2$ is of course trivial.) 
It is known that, to achieve this, the message of Alice has $\Omega(n)$ bits in the worst case.
See~\cite{v004a006} for a short proof and~\cite{CC} for a comprehensive treatment.

Given an input $(S,i)$ for \Index, we can build a data stream of same-length
intervals $\II$ with the property that $\alpha(\II)\in \{ 2,3\}$ and  
\Index$(S,i)=1$ if and only if $\alpha(\II)=3$. Moreover, the first
part of the stream depends only on $S$ and the second part on $i$. Thus,
the state of the memory at the end of the first part
can be interpreted as a message that Alice sends to Bob. This implies the
following lower bound.

\begin{figure}
	\includegraphics[page=1,width=\textwidth]{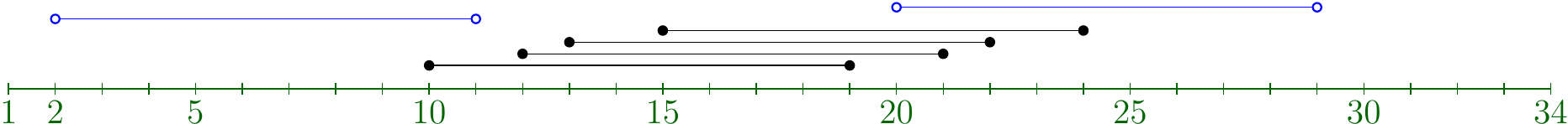}
    \caption{Example showing $\sigma(S,i)$ for $n=7$, $S=\{1,3,4,6\}$,
    	$L=9$, and $i=2$ in the proof of Theorem~\ref{thm:lowerbound1}.
    	The intervals are sorted from bottom to top in the order they appear 
        in the data stream.
        The empty dots represent endpoints that are not included in the interval, while
        the full dots represent endpoints included in the interval.}
    \label{fig:lowerbound1}
\end{figure}

\begin{theorem}
\label{thm:lowerbound1}
	Let $c>0$ be an arbitrary constant.
	Consider the problem of estimating $\alpha(\II)$ for sets of same-length intervals $\II$
	with endpoints in $[n]$.
	In the data streaming model,
	there is no algorithm that uses $o(n)$ bits of memory and computes
	an estimate $\hat \alpha$ such that 
	\begin{equation}\label{eq3}
		\Pr\left[ \left(\frac 23 + c\right) \alpha(\II) \le \hat\alpha \le \alpha(\II) \right] ~\ge~ \frac 23.
	\end{equation}
\end{theorem}
\begin{proof}
	For simplicity, we use intervals with endpoints in $[3n]$ and mix closed and open intervals in the proof.
    
	Given an input $(S,i)$ for \Index, consider the following stream of intervals. 
    We set $L$ to some large enough value; for example $L=n+2$ will be enough.
    Let $\sigma_1(S)$ be a stream that, for each $j\in S$, contains the closed interval $[L+j,2L+j]$.
    Let $\sigma_2(i)$ be the length-two stream with open intervals $(i,L+i)$ and $(2L+i,3L+i)$.
    Finally, let $\sigma(S,i)$ be the concatenation of $\sigma_1(S)$ and $\sigma_2(i)$.
    See Figure~\ref{fig:lowerbound1} for an example.
    Let $\II$ be the intervals in $\sigma(S,i)$.
    It is straightforward to see that $\alpha(\II)$ is $2$ or $3$.
    Moreover, $\alpha(\II)=3$ if and only if \Index$(S,i)=1$.

	Assume, for the sake of contradiction, that we have an algorithm in the data streaming
	model that uses $o(n)$ bits of space and computes a value $\hat\alpha$ 
	that satisfies equation~\eqref{eq3}.
	Then, Alice and Bob can solve \Index$(S,i)$ using $o(n)$ bits, as follows.
	
	Alice simulates the data stream algorithm on $\sigma_1(S)$ and sends to Bob
	a message encoding the state of the memory at the end of processing $\sigma_1(S)$. 
	The message of Alice has $o(n)$ bits.
	Then, Bob continues the simulation on the last two items of $\sigma(S,i)$, 
    that is, $\sigma_2(i)$. 
	Bob has correctly computed the output of the algorithm on $\sigma(S,i)$,
	and therefore obtains $\hat \alpha$ so that equation~\eqref{eq3} is satisfied.
	If $\hat\alpha > 2$, then Bob returns the bit $\hat\beta = 1$.
	If $\hat\alpha \le 2$, then Bob returns $\hat\beta=0$.
    This finishes the description of the protocol.
    
	Consider the case when \Index$(S,i)=1$. In that case, $\alpha(\II)=3$.
    With probability at least $2/3$, the value $\hat\alpha$ computed
    satisfies $\hat\alpha\ge (\tfrac 23 + c) \alpha(\II) = 2+3c>2$, and therefore
    \begin{eqnarray*}
    	\Pr\left[ \hat\beta=1\mid \text{\Index$(S,i)=1$}\right] 
    	& = & \Pr\left[ \hat\alpha>2\mid \text{\Index$(S,i)=1$}\right]\\
        & \ge & \Pr\left[ \hat\alpha\ge \left(\tfrac 23 + c\right) \alpha(\II)\mid \text{\Index$(S,i)=1$}\right]\\
        & \ge & 2/3.
    \end{eqnarray*}
	When \Index$(S,i)=0$, then $\alpha(\II)=2$, and,
    with probability at least $2/3$, the value $\hat\alpha$ computed
    satisfies $\hat\alpha\le \alpha(\II) = 2$. Therefore,
    \begin{eqnarray*}
    	\Pr\left[ \hat\beta=0\mid \text{\Index$(S,i)=0$}\right] 
    	& = & \Pr\left[ \hat\alpha \le 2\mid \text{\Index$(S,i)=0$}\right]\\
    	& \ge & \Pr\left[ \hat\alpha \le \alpha(\II)\mid \text{\Index$(S,i)=0$}\right]\\
    	& \ge & 2/3.
    \end{eqnarray*}
	We conclude that 
    \[
    	\Pr\left[ \hat\beta =\text{\Index}(S,i)\right] ~\ge~ 2/3.
    \]
	Since Bob computes $\hat\beta$ after a message from Alice with $o(n)$ bits,
    this contradicts the lower bound for \Index.
\end{proof}

For intervals of different sizes, we can use an alternative construction
with the property that $\alpha(\II)$ is either $k+1$ or $2k+1$. This means
that we cannot get an approximation ratio arbitrarily close to $2$.

\begin{figure}
	\includegraphics[page=2,width=\textwidth]{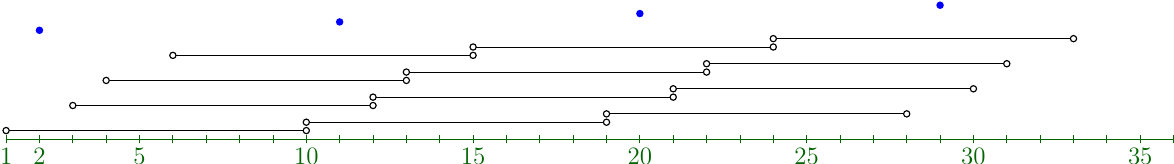}
    \caption{Example showing $\sigma(S,i)$ for $n=7$, $S=\{1,3,4,6\}$,
    	$L=9$, $k=3$, and $i=2$ in the proof of Theorem~\ref{thm:lowerbound2}.
    	The intervals are sorted from bottom to top in the order they appear 
        in the data stream.
        The empty dots represent endpoints that are not included in the interval, while
        the full dots represent endpoints included in the interval.}
    \label{fig:lowerbound2}
\end{figure}

\begin{theorem}
\label{thm:lowerbound2}
	Let $c>0$ be an arbitrary constant.
	Consider the problem of estimating $\alpha(\II)$ for sets of intervals $\II$
	with endpoints in $[n]$.
	In the data streaming model,
	there is no algorithm that uses $o(n)$ bits of memory and computes
	an estimate $\hat \alpha$ such that 
	\[
		\Pr\left[ \left(\frac 12 + c\right) \alpha(\II) \le \hat\alpha \le \alpha(\II) \right] ~\ge~ \frac 23.
	\]
\end{theorem}
\begin{proof}
	Let $k$ be a constant larger than $1/c$. 
    For simplicity, we will use
    intervals with endpoints in $[2kn]$.
    
	Given an input $(S,i)$ for \Index, consider the following stream of intervals. 
    We set $L$ to some large enough value; for example $L=n+2$ will be enough.
    Let $\sigma_1(S)$ be a stream that, for each $j\in S$, contains 
    the $k$ open intervals $(j,L+j), (j+L,j+2L),\dots, (j+(k-1)L,j+kL)$.
    Thus $\sigma_1(S)$ has exactly $k|S|$ intervals.
    Let $\sigma_2(i)$ be the stream with $k+1$ zero-length closed intervals 
    $[i,i], [i+L,i+L],\dots, [i+kL,i+kL]$.
    Finally, let $\sigma(S,i)$ be the concatenation of $\sigma_1(S)$ and $\sigma_2(i)$.
    See Figure~\ref{fig:lowerbound2} for an example.
    Let $\II$ be the intervals in $\sigma(S,i)$.
    It is straightforward to see that $\alpha(\II)$ is $k+1$ or $2k+1$:
    The greedy left-to-right optimum contains either all the intervals of $\sigma_2(i)$
    or those together with $k$ intervals from $\sigma_1(S)$.
    This means that $\alpha(\II)=2k+1$ if and only if \Index$(S,i)=1$.

	We use a protocol similar to that of the proof of Theorem~\ref{thm:lowerbound1}:
	Alice simulates the data stream algorithm on $\sigma_1(S)$, Bob receives
	the data message from Alice and continues the algorithm on $\sigma_2(S)$, and Bob returns
	bit $\hat\beta=1$ if an only if $\hat\alpha> k+1$.
	With the same argument, and using the fact 
	that $(\tfrac{1}{2}+c)(2k+1) > k+1$ by the choice of $k$, we can prove 
    that using $o(kn)=o(n)$ bits of memory one cannot distinguish, 
    with probability at least $2/3$, whether $\alpha(\II)=k+1$
    or $\alpha(\II)=2k+1$. The result follows.
\end{proof}

\small
\bibliographystyle{abuser}
\bibliography{biblio}

\end{document}

\end{document}